\documentclass[a4paper,UKenglish,cleveref,autoref,thm-restate]{arxiv}

\nolinenumbers

\usepackage{stmaryrd}
\usepackage{tikz-cd}

\definecolor{cornflowerblue}{rgb}{0.99,0.78,0.07}%{100,149,237}
\definecolor[named]{lipicsYellow}{rgb}{0.99,0.78,0.07}

\usepackage{mathtools}
\newcommand{\defeq}{\vcentcolon=}

\newcommand{\ket}[1]{  |{#1} \rangle} %{ \left |{#1}\right \rangle}
\newcommand{\bra}[1]{ \left \langle#1\right | }

\newcommand{\se}{\subseteq}
\newcommand{\ls}{\leqslant}
\newcommand{\gs}{\geqslant}
\newcommand{\sm}{\setminus}

\DeclareSymbolFont{matha}{OML}{txmi}{m}{it}% txfonts

\definecolor{fillvertices}{RGB}{250,240,230}
\newcommand{\supp}{\textup{\textsf{supp}}}
\newcommand{\ham}[1]{w(#1)}

\title{Local equivalence of stabilizer states: a graphical characterisation\footnote{Long version, including proofs, of the paper accepted at STACS 2025 \cite{CP25}.}} 

\author{Nathan Claudet}{Inria Mocqua, LORIA, CNRS, Université de Lorraine, F-54000 Nancy, France}{nathan.claudet@inria.fr}{https://orcid.org/0009-0000-0862-1264}{}

\author{Simon Perdrix}{Inria Mocqua, LORIA, CNRS, Université de Lorraine, F-54000 Nancy, France}{simon.perdrix@loria.fr}{https://orcid.org/0000-0002-1808-2409}{}


\authorrunning{N. Claudet and S. Perdrix}

\Copyright{Nathan Claudet and Simon Perdrix}

\ArticleNo{23} %no idea why but we need this to display page numbers

\begin{document}

\maketitle

\begin{abstract}  
    Stabilizer states form a ubiquitous family of quantum states that can be graphically represented through the graph state formalism. A fundamental property of  graph states is that applying a local complementation -- a well-known and extensively studied graph transformation -- results in a graph that represents the same entanglement as the original. In other words, the corresponding graph states are LU-equivalent. This property served as the cornerstone for capturing non-trivial quantum properties in a simple graphical manner, in the study of quantum entanglement but also for developing protocols and models based on graph states and stabilizer states, such as measurement-based quantum computing, secret sharing, error correction, entanglement distribution... However, local complementation fails short to fully characterise entanglement: there exist pairs of graph states that are LU-equivalent but cannot be transformed one into the other using local complementations. Only few is known about the equivalence of graph states beyond local complementation. We introduce a generalisation of local complementation which graphically characterises the LU-equivalence of graph states. We use this characterisation to show the existence of a strict infinite hierarchy of equivalences  of graph states. Our approach is based on minimal local sets, which are subsets of vertices that are known to cover any graph, and to be  invariant under local complementation and even LU-equivalence. We use these structures to provide a type to each vertex of a graph, leading to a natural standard form in which the LU-equivalence can be exhibited and captured by means of generalised local complementation. 
\end{abstract}

\section{Introduction}

Stabilizer states, and in particular graph states, form a versatile family of entangled quantum states, that allow for easy and compact representations. They are used as entangled resource states in various quantum information applications, like measurement-based computation \cite{raussendorf2001one, raussendorf2003measurement, briegel2009measurement}, error corrections \cite{schlingemann2001quantum,schlingemann2001stabilizer,cross2008codeword,sarvepalli2011local}, quantum communication network routing \cite{hahn2019quantum,meignant2019distributing, bravyi2024generating, Cautres2024}, and quantum secret sharing \cite{markham2008graph,gravier2013quantum},  to cite a few.  
In all these applications, stabilizer states are used as multi-partite entangled resources, it is thus crucial to understand when two such states have the same entanglement. 
According to standard quantum information theory, two quantum states have the same entanglement if they can be transformed into each other using only local operations, where `local' refers to operations that can be applied to at most one qubit at a time.  Indeed, intuitively, such local operations can only decrease the strength of the entanglement of a quantum state, thus if a state can be transformed into another and back to the original one using only local operations, the two states do have the same entanglement. Concretely, the most general case is the so-called SLOCC-equivalence (stochastic local operations and classical communications) that encompass the use of local unitaries and measurements. In the particular case of stabilizer states, it is enough to consider \emph{LU-equivalence} (local unitaries), as two stabilizer states are SLOCC-equivalent iff there exists $U=U_1\otimes \ldots \otimes U_n$ that transforms one state into the other, where each $U_i$ is a 1-qubit unitary transformation \cite{Hein04}. Therefore, in this paper, we refer to LU-equivalence when stating that two stabilizer states possess the same entanglement.

Stabilizer states can be graphically represented through the graph state formalism, which consists in representing quantum states using graphs where each vertex corresponds to a qubit.  Graph states form actually a sub-family of stabilizer states, but any stabilizer state can be transformed into a graph state using local Clifford unitaries in a straightforward way \cite{Hein04}, as a consequence it is sufficient to consider graph states when dealing with entanglement equivalence of stabilizer states. Notice that understanding the entanglement structure of graph states and how it can be characterised in a graphical way is the most fascinating fundamental question of the graph state formalism \cite{Hein04,VandenNest04,VandenNest05,gross2007lu,Ji07,Tsimakuridze17,tzitrin2018local,DWH:transfo}.

An interesting combinatorial property of entanglement has been proved in \cite{Hein04}: LU-equivalent graph states have the same cut-rank function\footnote{The cut-rank function over a graph $G=(V,E)$ is the function that maps a set $A \se V$ to the rank (in $\mathbb F_2$) of the cut-matrix $\Gamma_A = ((\Gamma_A)_{ab}: a \in A\text{, } b \in V\sm A)$ where $\Gamma_{ab} = 1$ if and only if $(a, b) \in E$.}. However, there exist pairs of graph states that have the same cut-rank function, but that are not LU-equivalent: a counterexample involves two isomorphic Petersen graphs~\cite{Fon-Der-Flaass1996,Hein04}.

A fundamental property of the graph state formalism is that a \emph{local complementation} -- a graph transformation that consists in complementing the neighbourhood of a given vertex -- preserves entanglement: if a graph can be transformed into another  by means of local complementations then the two corresponding graph states are LU-equivalent. More precisely the corresponding graph states are LC-equivalent (local Clifford). Van den Nest \cite{VandenNest04} proved that two graph states are LC-equivalent if and only if the corresponding graphs are related by a sequence of local complementations. LC-equivalence and LU-equivalence of graph states were conjectured to coincide \cite{5pb}, as they coincide for several families of graph states \cite{VandenNest05,Zeng07}. However, a counterexample of order 27 has been discovered using computer assisted methods \cite{Ji07}. Since then, {the existence of a graphical characterisation of LU-equivalence has remained an open question.} Furthermore, families of graph states for which LC-equivalence and LU-equivalence coincide have been a subject of interest \cite{sarvepalli2011local, tzitrin2018local}.

\noindent{\bf Contributions.~} We introduce the \emph{$r$-local complementation}, a graphical transformation parametrised by a positive integer $r$, which coincides with the usual local complementation when $r=1$. We show that an $r$-local complementation can be performed on the corresponding graph states by means of local unitaries in LC$_r = \left< H,~Z\left(\pi/2^r\right) \right>$, the group generated by the local Clifford group augmented with the $Z\left(\pi/2^r\right)$ gate. We additionally show that, conversely, for any fixed $r$, the LC$_r$-equivalence of graph states is captured by $r$-local complementations: if two graph states are LC$_r$-equivalent, then there is a sequence made of (usual) local complementations and at most one $r$-local complementation, that transforms one graph into the other. In particular, the known examples of graph states that are LU-equivalent but not LC-equivalent \cite{Ji07,Tsimakuridze17}, are actually LC$_2$-equivalent, thus they are captured graphically by the $2$-local complementation. 
This leads to the natural question of the existence of pairs of graph states that are LC$_3$-equivalent but not LC$_2$-equivalent, and so on. To answer this question, we show that the $r$-local complementations form a strict hierarchy: for any positive $r$ there exist pairs of LC$_{r+1}$-equivalent graph states that are not LC$_r$-equivalent. Finally, we show that any LU-equivalent graph states are LC$_r$-equivalent for some $r$, even if there exist some local unitaries, like $R_Z(\pi/3)$, that are not contained in LC$_r$, for any $r$. In other words, two graph states are LU-equivalent if and only if the corresponding graphs can be transformed into each other by a sequence of generalised local complementations, leading to a full graphical description of LU-equivalence for graph states. 
As an application, we prove a conjecture on a family of graph states for which LU-equivalence and LC-equivalence coincide. 

\noindent{\bf Related work.} Other graphical extensions of local complementation have been introduced in particular for weighted hypergraph states \cite{Tsimakuridze17}, a wider family of quantum states. The connection between generalised local complementation and local complementation of weighted hypergraph states is discussed in section \ref{sec:r_lc}.  In \cite{gross2007lu,zeng2011transversality}, the authors proved that when considering LU-equivalence of stabilizer states it suffices to consider semi-Clifford unitaries\footnote{A unitary $U$ is semi-Clifford if there exists a Pauli operator $P \in {X, Y, Z}$ such that $UPU^\dagger$ is also a Pauli operator.}. Symmetries of stabilizers and graph states have been explored in \cite{Englbrecht2020}, where a characterisation of local operations that have graph states as fixed points, is provided; the characterisation of the LU-equivalence of graph states was however left as an open question. Shortly after the submission\footnote{The present paper is the long version, including proofs, of the paper accepted at STACS 2025 \cite{CP25}.} of the present paper and its upload on arXiv, Burchardt, de Jong, and Vandré have introduced an algorithm to decide the LU-equivalence of stabilizer states \cite{burchardt2024algorithmverifylocalequivalence}. This independent work uses some similar techniques to ours,  notably minimal local sets. 

\noindent{\bf Structure of the paper.~} In \cref{sec:preliminaries}, we recall some definitions and notations and provide a handful of tools to manipulate graph states. In \cref{sec:r_lc}, we define the generalisation of the local complementation, and we show that $r$-local complementations can be implemented on the corresponding graph states by means of local unitaries in LC$_r$. 
\cref{sec:standardform} is dedicated to the converse result: the graphs corresponding to any pair of LC$_r$-equivalent graph states can be transformed into each other by means of $r$-local complementations. To do so, we show that any graph can be put in a standard form by means of (usual) local complementations. This standard form is based on the so-called minimal local sets, which are subsets of vertices known to be invariant under LU-equivalence. As any graph can be covered by minimal local sets, we associate to any vertex of a graph a type on which is based the standard form. We then prove that if two graph states in standard form are LC$_r$-equivalent, there exists an $r$-local complementation transforming one graph into the other. We also prove that if two graph states on $n$ qubits are LU-equivalent they are LC$_{n/2}$-equivalent. 
Finally, in \cref{sec:hierarchy}, we introduce a family of variants of Kneser graphs, and show, using our graphical characterisation of LC$_r$-equivalence, that for any $r>1$ there exists pairs of graph states that are LC$_r$-equivalent but not LC$_{r-1}$-equivalent.

\section{Preliminaries}
\label{sec:preliminaries}

Let us first give some notations and basic definitions. 
A graph\footnote{We only consider simple (no selfloop) undirected graphs.} $G$ is a pair $(V,E)$, where $V$ is a finite set of vertices, and $E$ is a set of {unordered pairs of distinct vertices called} 
edges. We assume the set of vertices to be totally ordered and use the notation $V=\{u_1,\ldots, u_n\}$ s.t. $u_i\prec u_j$ iff $i<j$. The number $n=|V|$ of vertices is called the order of the graph and $|G|:= |E|$ its size. We use the notation $u\sim_G v$ when the vertices $u$ and $v$ are adjacent in $G$, i.e. $(u,v)\in E$. Given a vertex $u$,  $N_G(u)=\{v\in V~|~(u,v)\in E\}$ is the neighbourhood of $u$. The odd and the common neighbourhoods are two natural generalisations of the notion  of neighbourhoods  to sets of vertices: for any $D \se V$, $Odd_G(D) =  \Delta_{u\in D} N_G(u) = \{v \in V ~|~|N_G(v) \cap D| = 1 \text{ mod } 2\}$ is the odd neighbourhood of $D$, where $\Delta$ denotes the symmetric difference on vertices. Informally, $Odd_G(D)$ is the set of vertices that are the neighbours of an odd number of vertices in $D$. The common neighbourhood of $D$ is denoted $\Lambda_G^D=\bigcap_{u\in D}N_G(u)=\{v\in V~|~\forall u \in D, v\in N_G(u)\}$. Two unadjacent vertices $u$ and $v$ are said twins if they are adjacent to the same vertices i.e. $N_G(u) = N_G(v)$. A local complementation with respect to a given vertex $u$ consists in complementing the subgraph induced by the neighbourhood of $u$, leading to the graph $G\star u= G\Delta K_{N_G(u)}$ where $\Delta$ denotes the symmetric difference on edges and $K_A$ is the complete graph on the vertices of $A$. Local complementation is an involution, i.e. $G \star u \star u = G$. A pivoting with respect to two adjacent vertices $u$ and $v$ is the operation that maps $G$ to $G\wedge uv:= G\star u \star v \star u$. Notice that pivoting is symmetric ($G\wedge uv = G\wedge vu$) and is an involution ($G\wedge uv\wedge uv = G$). With a slight abuse of notation we identify multisets of vertices with their multiplicity function $V\to \mathbb N$. \footnote{Hence, we also identify sets of vertices with their indicator functions $V\to \{0,1\}$.}
A (multi)set $S$ of vertices  is independent if there is no two vertices of $S$ that are adjacent.

Graph states form a standard family of quantum states that can be represented using graphs (Ref. \cite{Hein06} is an excellent introduction to graph states). Given a graph $G$ of order $n$, the corresponding \textbf{graph state} $\ket G$ is the $n$-qubit state: $$\ket G = \frac 1{\sqrt {2^n}}\sum_{x\in \{0,1\}^n}(-1)^{|G[x]|}\ket x$$ where $G[x]$ denotes the subgraph of $G$ induced by $\{u_i~| ~x_i=1\}\subseteq V$. 

The $i^\text{th}$ qubit of $\ket G$ can be associated with the vertex $u_i$ of $G$, so, with a slight abuse of notation, we refer to $V$ as the set of qubits of $\ket G$. Given a 1-qubit operation $R$, like $X:\ket a\mapsto \ket {1-a}$ or $Z:\ket a \mapsto (-1)^a\ket a$, and a subset $D\subseteq V$ of qubits, $R_D:=\bigotimes_{u\in D} R_u$ is the local operation consisting in applying $R$ on each qubit of $D$.

The graph state $\ket G$ is the unique quantum state (up to a global phase) that, for every vertex $u\in V$, is a fixed point of the Pauli operator $X_uZ_{N_G(u)}$. 
More generally, $\ket G$ is an eigenvector of any Pauli operator $X_DZ_{Odd(D)}$ with $D\subseteq V$:

\begin{restatable}{proposition}{stabilizer}
    Let $G = (V,E)$ be a graph. The corresponding graph state $\ket G$ is a fixed point of $\mathcal P^G_D=(-1)^{|G[D]|}X_D Z_{Odd(D)}$ for any $D\subseteq V$,
    where $|G[D]|$ is the number of edges of the subgraph of $G$ induced by $D$.
    \label{prop:stabilizer}
\end{restatable}
    
The proof of \cref{prop:stabilizer} is given in Appendix \ref{app:proof:stabilizer}. 
The fact that a graph state can be defined as the common eigenvector of Pauli operators witnesses that graph states form a subfamily of stabilizer states\footnote{An $n$-qubit state is a stabilizer state if it is the fixpoint of $n$ independent commuting Pauli operators (or equivalently $2^n$ distinct commuting Pauli operators).}. Conversely, any stabilizer state can be turned into a graph state by means of local Clifford unitaries\footnote{Local Clifford unitaries are tensor products of single-qubit Clifford gates, which map the Pauli group to itself under conjugation. Formally, a single-qubit gate $U$ is Clifford if for any $P \in \mathcal P \defeq \{\pm 1, \pm i\} \times \{I, X, Y, Z \}$, $U P U^\dagger \in \mathcal P$.}  in a straightforward way (see for instance \cite{VandenNest04}). Thus, we will focus in the rest of this paper on graph states. 

The action of measuring, in the standard basis, a qubit of a graph state leads, roughly speaking, to the graph state where the corresponding vertex has been removed. A diagonal measurement can also be used to remove a vertex when this vertex is isolated: 

\begin{proposition}\label{prop:meas}
For any graph $G$ and any vertex $u$,  $\bra 0_u \ket G = \frac1{\sqrt 2} \ket{G\setminus u}$. Moreover,  if $u$ is an isolated vertex (i.e. $N_G(u)=\emptyset$), then $\bra +_u\ket G = \ket{G\setminus u}$ 
where $\bra + = \frac{\bra 0 + \bra 1}{\sqrt 2}$. 
\end{proposition}

\begin{remark}
A standard basis measurement consists in applying either $\bra 0$ or $\bra 1$ which correspond the classical outcomes $0$ and $1$ respectively. Whereas the action in the $0$ case is directly described in \cref{prop:meas}, the action in the $1$ case can be recovered thanks to \cref{prop:stabilizer}: $\bra 1_u\ket G = \bra 0_u X_u\ket G = \bra 0_u Z_{N_G(u)}\ket G =  \frac1{\sqrt 2}Z_{N_G(u)}\ket{G\setminus u}$. Thus, it also corresponds to a vertex deletion up to some $Z$ corrections on the neighbourhood of the measured qubit. 
\end{remark}

We are interested in the action of 1-qubit unitaries on graph states, in particular Hadamard $H:\ket{a}\mapsto \frac{\ket 0+(-1)^a\ket1}{\sqrt 2}$, and Z- and X-rotations respectively defined as follows:
\[Z(\alpha):=e^{i\frac \alpha2}\left(\cos\left(\frac \alpha2\right)I-i\sin\left(\frac \alpha2\right)Z\right)\quad  X(\alpha):=HZ(\alpha)H=e^{i\frac \alpha2}\left(\cos\left(\frac \alpha2\right)I-i\sin\left(\frac \alpha2\right)X\right)\]

In particular, local complementations can be implemented by $\pi/2$ X- and Z-rotations:

\begin{proposition}[\cite{VandenNest04}]For any graph $G=(V,E)$ and any $u\in V$,  $\ket{G\star u} = L_u^G \ket {G}$ where $L_u^G := X_u(\frac \pi 2)\bigotimes_{v\in N_G(u)}Z_v(-\frac \pi 2)$. 
\end{proposition}

Similarly, pivoting can be implemented by means of Hadamard transformations (up to Pauli transformations):
\begin{proposition}[\cite{mhalla2012graph,van2005edge}]For any graph $G=(V,E)$ and any $(u,v)\in E$, 
$\ket{G\wedge uv} = H_uH_vZ_{\Lambda_{G}^{\{u,v\}}}\ket {G}$ 
\end{proposition}

Beyond $\frac \pi 2$ rotations, notice that any local unitary can be implemented using Hadamard and Z-rotations with arbitrary angles. We consider for any $r\in \mathbb N$ the set LC$_r$ of local unitaries generated by $H$ and $Z(\frac\pi{2^r})$. In particular LC$_1$ is the set of local Clifford operators, moreover, for any $r$, LC$_r\subseteq LC_{r+1}$. Notice that there exist local unitaries, e.g. $Z(\frac \pi3)$, that are not contained in  LC$_r$ for any $r$. We say that two states are LU-equivalent, denoted $\ket{\psi} =_{LU} \ket{\psi'}$ when there exists a local unitary transformation $U$ such that $\ket \psi = U\ket {\psi'}$. Similarly, two states are LC$_r$-equivalent, denoted $\ket\psi=_{LC_r} \ket{\psi'}$ when there exists a local unitary $U\in LC_r$ s.t. $\ket{\psi} = U\ket{\psi'}$.

It is well-known that two graph states $\ket G$ and $\ket{G'}$ are LC$_1$-equivalent if and only if there exists a sequence of local complementations that transforms $G$ into $G'$ \cite{VandenNest04}. We slightly refine this result showing that there exists a reasonably small sequence of local complementation transforming $G$ into $G'$, that corresponds to a particular Clifford operator:

\begin{restatable}{proposition}{LCLC}
If $\ket {G_2}= C\ket {G_1}$ where $C$ is a local Clifford operator up to a global phase, then there exists a sequence of (possibly repeating) vertices $a_1, \cdots, a_m$ such that $G_2 = G_1\star a_1 \star a_2 \star \cdots \star a_{m}$ with $m \ls \left\lfloor 3n/2\right\rfloor$, and the local Clifford operator that implements these local complementations is $C$ up to a Pauli operator.
\label{prop:LCLC}
\end{restatable}

The proof of \cref{prop:LCLC} is given in Appendix \ref{app:proof:LCLC}. 

\begin{remark}\label{rm:clifford}
The induced Clifford operator is of the form $\mathcal P_D^GC$ for some subset $D$ of vertices. As $X_u Z_{N_G(u)} = (L_{u}^G)^2$, there exists a sequence at most $\lfloor 7n/2 \rfloor$ local complementations that transforms $G_1$ into $G_2$ and induces exactly the local Clifford operator $C$.  
\end{remark}

Whereas it has been originally conjectured that if two graphs states are LU-equivalent then there exists a sequence of local complementations transforming one in another \cite{5pb}, a couple of counter examples of pairs of LC$_2$- but not LC$_1$-equivalent graph states have been pointed out \cite{Ji07, Tsimakuridze17}. To further understand  the action of local unitaries on graph states, we thus introduce in the next section a generalisation of local complementation to graphically capture  the equivalence of graph states that are LU- but not LC$_1$-equivalent.

\section{Generalising local complementation}
\label{sec:r_lc}

\subsection{Local complementation over independent sets}
As a first step towards a generalisation of  local complementation, we introduce a natural shortcut $G  \star S:=G\star a_1\ldots \star a_k$ to denote the graph obtained after a series of local complementations when $S=\{a_1, \ldots, a_k\}$ is an independent set in $G$. %It is worse noticing that 
The independence condition  is important as local complementations do not commute in general when applied on adjacent vertices.

Notice that the action of a local complementation over $S$ can be described directly as toggling edges that have an odd number of common neighbours in $S$:  
\begin{equation}\label{eq:LCIS}u\sim_{G\star S} v ~\Leftrightarrow~\left(u\sim_{G} v ~~\oplus~~ |S\cap N_G(u)\cap N_G(v)| = 1\bmod 2\right)\end{equation}

\subsection{Towards a 2-local complementation}

We introduce 2-local complementations as a refinement of \emph{idempotent} local complementations,~i.e.~when $G\star S=G$. According to \cref{eq:LCIS}, an idempotent local complementation occurs when each pair $(u,v)$ of vertices has an even number of common neighbours in $S$, one can then consider a new graph transformation consisting in toggling an edge $(u,v)$ when its number of common neighbours in $S$ is equal to 2 modulo 4. However, such an action may not be implementable using local operations in the graph state formalism (see \cref{fig:k4}). To guarantee an implementation by means of local unitary transformations on the corresponding graph states, we add the condition, called $2$-incidence, that any set of 2 or 3 (distinct) vertices has an even number of common neighbours in $S$.

\begin{figure}[t]
\centering

\scalebox{0.9}{
 \begin{tikzpicture}[scale = 0.4]        
        \begin{scope}[every node/.style={circle,minimum size=15pt,thick,draw,fill=lipicsYellow, inner sep = 0pt}]
            \node (e) at (-4,0) {$e$};
            \node (f) at (-1,0) {$f$};
            \node (g) at (2,0) {$g$};
            \node (h) at (5,0) {$h$};
            \node (a) at (-4,5) {$a$};
            \node (b) at (-1,5) {$b$};
            \node (c) at (2,5) {$c$};
            \node (d) at (5,5) {$d$};
         \end{scope}
        \begin{scope}[every node/.style={},
                        every edge/.style={draw=darkgray,very thick}]              
            \path [-] (a) edge node {} (f);
            \path [-] (a) edge node {} (g);
            \path [-] (a) edge node {} (h);
            \path [-] (b) edge node {} (e);
            \path [-] (b) edge node {} (g);
            \path [-] (b) edge node {} (h);
            \path [-] (c) edge node {} (e);
            \path [-] (c) edge node {} (f);
            \path [-] (c) edge node {} (h);
            \path [-] (d) edge node {} (e);
            \path [-] (d) edge node {} (f);
            \path [-] (d) edge node {} (g);

        \end{scope}
\end{tikzpicture}\qquad\qquad\qquad\qquad
\raisebox{-0.7cm}{
    \begin{tikzpicture}[scale = 0.4]        
        \begin{scope}[every node/.style={circle,minimum size=15pt,thick,draw,fill=lipicsYellow, inner sep = 0pt}]
            \node (e) at (-4,0) {$e$};
            \node (f) at (-1,0) {$f$};
            \node (g) at (2,0) {$g$};
            \node (h) at (5,0) {$h$};
            \node (a) at (-4,5) {$a$};
            \node (b) at (-1,5) {$b$};
            \node (c) at (2,5) {$c$};
            \node (d) at (5,5) {$d$};
         \end{scope}
        \begin{scope}[every node/.style={},
                        every edge/.style={draw=darkgray,very thick}]              
            \path [-] (a) edge node {} (f);
            \path [-] (a) edge node {} (g);
            \path [-] (a) edge node {} (h);
            \path [-] (b) edge node {} (e);
            \path [-] (b) edge node {} (g);
            \path [-] (b) edge node {} (h);
            \path [-] (c) edge node {} (e);
            \path [-] (c) edge node {} (f);
            \path [-] (c) edge node {} (h);
            \path [-] (d) edge node {} (e);
            \path [-] (d) edge node {} (f);
            \path [-] (d) edge node {} (g);

            \path [-] (e) edge node {} (f);
            \path [-] (e) edge [bend right]  node {} (g);
            \path [-] (e) edge [bend right=40]  node {} (h);
            \path [-] (f) edge  node {} (g);
            \path [-] (f) edge[bend right]  node {} (h);
            \path [-] (g) edge node {} (h);
        \end{scope}
\end{tikzpicture}}
}
\caption{\label{fig:k4}Example of non LU-equivalent graph states. 
With the notations of \cref{sec:hierarchy}, the graph on the left is a $C_{4,3}$ graph, and the one on the right is a $C'_{4,3}$ graph. Notice that applying a local complementation on the upper part of the bipartition ($S=\{a,b,c,d\}$) leaves the graphs invariant as each pair of vertices on the other part (e.g. $(e,f)$), has two common neighbours in $S$. However, a 2-local complementation over $S$ is not valid in both graphs as $S$ is not $2$-incident. See \cref{subsec:mls} for a proof that these graph states are not LU-equivalent.} 
\end{figure}
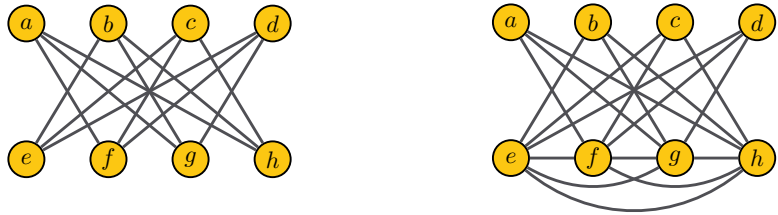

For instance in the following example, a 2-local complementation over $\{a,b\}$ performs the following transformation:

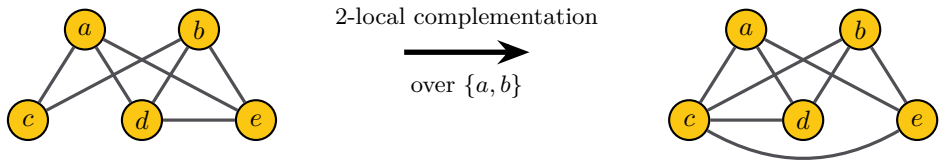
\[
\begin{tikzpicture}[scale = 0.3]

\begin{scope}[every node/.style={circle,minimum size=15pt,thick,draw,fill=lipicsYellow, inner sep = 0pt}]
    \node (a) at (-2.5,4) {$a$};
    \node (b) at (2.5,4) {$b$};
    \node (c) at (-5,0) {$c$};
    \node (d) at (0,0) {$d$};
    \node (e) at (5,0) {$e$};
\end{scope}
\begin{scope}[every node/.style={},
                every edge/.style={draw=darkgray,very thick}]                   
    \path [-] (a) edge node {} (c);
    \path [-] (a) edge node {} (d);
    \path [-] (a) edge node {} (e);
    \path [-] (b) edge node {} (c);
    \path [-] (b) edge node {} (d);
    \path [-] (b) edge node {} (e);
    \path [-] (d) edge node {} (e);
\end{scope}

\begin{scope}[shift={(29,0)},every node/.style={circle,minimum size=15pt,thick,draw,fill=lipicsYellow, inner sep = 0pt}]
    \node (a) at (-2.5,4) {$a$};
    \node (b) at (2.5,4) {$b$};
    \node (c) at (-5,0) {$c$};
    \node (d) at (0,0) {$d$};
    \node (e) at (5,0) {$e$};
\end{scope}
\begin{scope}[every node/.style={},
                every edge/.style={draw=darkgray,very thick}]                   
    \path [-] (a) edge node {} (c);
    \path [-] (a) edge node {} (d);
    \path [-] (a) edge node {} (e);
    \path [-] (b) edge node {} (c);
    \path [-] (b) edge node {} (d);
    \path [-] (b) edge node {} (e);
    \path [-] (d) edge node {} (c);
     \path [-] (c) edge [bend right]node {} (e);
\end{scope}

\draw[-Stealth, line width=2pt] (11.5,3) -- (17,3);
\node (t) at (14.25,4.5) {\small 2-local complementation};
\node (t2) at (14.25,1.5) {\small  over $\{a,b\}$};   

\end{tikzpicture}\]
In the LHS graph $G$, $S=\{a,b\}$ is an independent set. Moreover, a $1$-local complementation over $S$ is idempotent: $G\star S = G$, as $a$ and $b$ are twins. $S$ is also $2$-incident as any pairs and triplets of vertices have an even number of neighbours in $S$. Thus, a $2$-local complementation over $S$ is valid and consists in toggling the edges $(c,d)$, $(c,e)$, and $(d,e)$ as each of them has a number of common neighbours in $S$ equal to $2 \bmod 4$. 

We also consider the case where $S$ is actually a multiset, like in \cref{fig:generalized_lc}  where a 2-local complementation over %at \Scom{over?} 
the multiset $\{a,a,b,c\}$ is performed. When $S$ is a multiset, the number of common neighbours in $S$ should be counted with their multiplicity. 

Notice that when a 2-local complementation is invariant one can similarly refine the 2-local complementation into a 3-local complementation, leading to the general definition of generalised local complementation provided in the next section.

\subsection{Defining generalised local complementation}

We introduce a generalisation of local complementation, that we call $r$-local complementation for any positive integer $r$. This generalised local complementation denoted $G\star^r S$ is parametrised by a (multi)set $S$, that has to be independent and also  $r$-\emph{incident}, which is the following counting condition on the number of common neighbours in $S$ of any set that does not intersect $\supp(S)$,  the support of $S$: 

\begin{definition}[$r$-Incidence]
Given a graph $G$, a multiset $S$ of vertices is  $r$-incident, if for any $k\in [0,r)$, and any $K\subseteq V\setminus \supp(S)$ of size $k+2$, $S\bullet \Lambda_G^K$ is a multiple of $2^{r-k-\delta(k)}$,
where $\delta$ is the Kronecker delta\footnote{$\delta(x)\in \{0,1\}$ and $\delta(x)=1 \Leftrightarrow x=0$.
} and $S\bullet \Lambda_G^K$ is the number of vertices of $S$, counted with their multiplicity, that are neighbours to all vertices of $K$.\footnote{$.\bullet.$ is the scalar product: $A\bullet B = \sum_{u\in V}A(u).B(u)$, so $S\bullet \Lambda_G^K = \sum_{u \in \Lambda_G^K}S(u)$.}   
\end{definition}

\begin{definition}[$r$-Local Complementation]
Given a graph $G$ and an $r$-incident independent multiset $S$, let $G\star^rS$ be the graph defined as
\[u\sim_{G\star^r S} v ~\Leftrightarrow~\left(u\sim_{G} v ~~\oplus~~ S \bullet\Lambda_G^{u,v} = 2^{r-1}\bmod 2^{r}\right)\]
\end{definition}

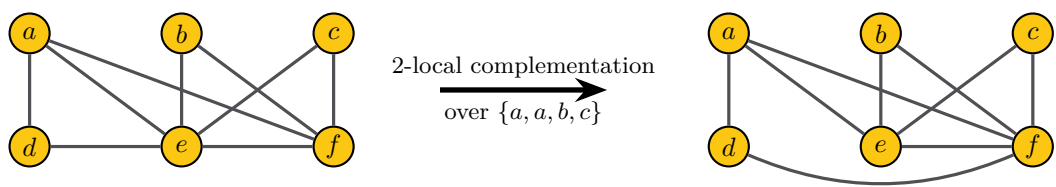
\begin{figure}[t]
\centering
\begin{tikzpicture}[xscale = 0.4,yscale=0.3]

\begin{scope}[every node/.style={circle,minimum size=15pt,thick,draw,fill=lipicsYellow, inner sep = 0pt}]
    \node (U1) at (-5,5) {$a$};
    \node (U2) at (0,5) {$b$};
    \node (U3) at (5,5) {$c$};
    \node (U4) at (-5,0) {$d$};
    \node (U5) at (0,0) {$e$};
    \node (U6) at (5,0) {$f$};
\end{scope}
\begin{scope}[every node/.style={},
                every edge/.style={draw=darkgray,very thick}]                   
    \path [-] (U1) edge node {} (U4);
    \path [-] (U1) edge node {} (U5);
    \path [-] (U1) edge node {} (U6);
    \path [-] (U2) edge node {} (U5);
    \path [-] (U2) edge node {} (U6);
    \path [-] (U3) edge node {} (U5);
    \path [-] (U3) edge node {} (U6);

    \path [-] (U4) edge node {} (U5);
    \path [-] (U5) edge node {} (U6);
\end{scope}

\begin{scope}[shift={(23,0)},every node/.style={circle,minimum size=15pt,thick,draw,fill=lipicsYellow, inner sep = 0pt}]
    \node (U1) at (-5,5) {$a$};
    \node (U2) at (0,5) {$b$};
    \node (U3) at (5,5) {$c$};
    \node (U4) at (-5,0) {$d$};
    \node (U5) at (0,0) {$e$};
    \node (U6) at (5,0) {$f$};
\end{scope}
\begin{scope}[every node/.style={},
                every edge/.style={draw=darkgray,very thick}]                   
    \path [-] (U1) edge node {} (U4);
    \path [-] (U1) edge node {} (U5);
    \path [-] (U1) edge node {} (U6);
    \path [-] (U2) edge node {} (U5);
    \path [-] (U2) edge node {} (U6);
    \path [-] (U3) edge node {} (U5);
    \path [-] (U3) edge node {} (U6);

    \path [-] (U4) edge[bend right] node {} (U6);
    \path [-] (U5) edge node {} (U6);

\end{scope}

\draw[-Stealth, line width=2pt] (8.5,2.5) -- (14,2.5);
\node (t) at (11.25,3.5) {\small 2-local complementation};
\node (t2) at (11.25,1.5) {\small  over $\{a,a,b,c\}$};

\end{tikzpicture}
\caption{Illustration on a 2-local complementation on the multiset $S = \{a,a,b,c\}$. $S$ is 2-incident: indeed $S\bullet \Lambda_G^{\{d,e,f\}} = 2$, which is a multiple of $2^{2-1-0} = 2$. Similarly, $S\bullet \Lambda_G^{\{d,e\}} = S\bullet \Lambda_G^{\{d,f\}} = 2$ and $S\bullet \Lambda_G^{\{e,f\}} = 4$. Edges $de$ and $df$ are toggled as $S\bullet \Lambda_G^{\{d,e\}} = S\bullet \Lambda_G^{\{d,f\}} = 2\bmod 4$, but not edge $ef$ as $S\bullet \Lambda_G^{\{e,f\}} = 0\bmod 4$.}
\label{fig:generalized_lc}    
\end{figure}

A detailed example of a $2$-local complementation is given in \cref{fig:generalized_lc}. 

We will also use $r$-local complementation parametrised with a set (rather than a multiset), in this case $S\bullet \Lambda_G^K = |S\cap \Lambda_G^K|$, and $ S \bullet\Lambda_G^{u,v} = |S\cap N_G(u)\cap N_G(v)|$. Note that we recover the vanilla local complementation when $r=1$.

We say that $G\star^r S$ is valid when $S$ is an $r$-incident independent multiset in $G$. Moreover, generalising the notion of local equivalence on graphs, we say that two graphs $G$, $G'$ are $r$-locally equivalent when there exists a sequence of $r$-local complementations transforming $G$ into $G'$. In the case of the usual local complementation, we simply say that the graphs are locally equivalent. 

Generalised local complementations satisfy a few basic properties (proven in \cref{app:prop_r_lc}), in particular, it is easy to double-check that they are self inverse: $(G\star^r S)\star^r S =G$. 
Moreover, $r$-local complementations can be related to $(r+1)$ and $(r-1)$-local complementations: 

\begin{restatable}{proposition}{rlcmonotonicity}
\label{prop:monotonicity}
If $G\star^r S$ is valid then:
\begin{itemize}
\item  $G\star^{r+1} (2S)$ is valid and induces the same transformation: 
$G\star^{r+1} (2S) = G\star^rS$, where $2S$ is the multiset obtained from $S$ by doubling the multiplicity of each vertex,
\item $G\star^{r-1} S$ is valid (when $r>1$) and  $G\star^{r-1} S = G$. 
\end{itemize}
\end{restatable}

It implies in particular that  two $r$-locally equivalent graphs are $(r+1)$-locally equivalent. An $r$-local complementation over $S$ preserves the neighbourhood of the vertices in $\supp(S)$:

\begin{restatable}{proposition}{rlcsameneighborhood}
    \label{loccompneighS}
If $G\star^r S$ is valid, then for any $u\in \supp(S)$, $N_{G\star^r S}(u)=N_G(u)$. 
\end{restatable}

Notice that $r$-local complementations can be composed:

\begin{restatable}{proposition}{rlcdisjointunion}
If $G\star^r S_1$ and $G\star^r S_2$ are valid and the disjoint union\footnote{For any vertex $u$, $(S_1\sqcup S_2)(u) = S_1(u)+S_2(u)$.}  $S_1\sqcup S_2$ is independent in $G$, then $G\star^r (S_1\sqcup S_2) = (G\star^r S_1)\star^r S_2$. 
\end{restatable}

Finally, it is easy to see that the multiplicity in $S$ can be upperbounded by $2^r$: if $G\star^r S$ is valid, then $G\star^r S= G\star^r S'$, where, for any vertex $u$, $S'(u)=S(u)\bmod 2^r$. 

\subsection{Implementation of generalised local complementation}

It is well known that local complementations can be implemented on graph states by means of local Clifford operations  \cite{Hein04,Hein06}, we extend this result to $r$-local complementations that can be implemented using $\frac\pi{2^r}$ X- and Z-rotations: 

\begin{restatable}{proposition}{rlcwithclifford}\label{prop:clif}
Given a graph $G=(V,E)$ and an $r$-incident independent multiset $S$ of vertices, \begin{equation*}\label{eq:LCr}\ket{G\star^r S} = \bigotimes_{u\in V}X\left(\frac {S(u)\pi}{2^r}\right)\bigotimes_{v\in V}Z\left(-\frac {\pi}{2^r}\sum_{u \in N_G(v)}S(u)\right)\ket{G}\end{equation*}
\end{restatable}

The proof of \cref{prop:clif} is given in Appendix \ref{sec:proof:clif}. 

\begin{remark}The graph state formalism has been extended in various ways, notably through the introduction of \emph{hypergraph states} \cite{rossi2013quantum} and even \emph{weighted hypergraph states} \cite{Tsimakuridze17}. Notice that the formalism of weighted hypergraph states provides a way to decompose a generalised local complementation over %at\Scom{over} 
a multiset $S$ of size $k$ into $k$ elementary transformations on weighted hypergraph states. 
Indeed, as shown in \cref{prop:clif}, a generalised local complementation can be implemented by means of $k$ X-rotations (along with some Z-rotations), moreover  the action of each X-rotation on weighted hypergraph states has been described in \cite{Tsimakuridze17}. However, in such a decomposition, the intermediate weighted hypergraph states are not necessarily graph states.
\end{remark} 

In the next section, we prove the converse result: if a local unitary operation in LC$_r$ transforms a graph state $\ket G$ into a graph state $\ket {G'}$ then there exists a sequence of $r$-local complementations transforming $G$ into $G'$. More so, we show that $r$-local complementation actually captures completely the LU-equivalence of graph states. To prove these results, we make use of graphical tools, namely the so-called minimal local sets and a standard form for graphs.

\section {Graphical characterisation of the action of local unitaries}

\label{sec:standardform}

\subsection{Minimal local sets and types}

\label{subsec:mls}

To characterise the action of local unitaries on graph states, we rely on properties that are invariant under local unitaries. A local set \cite{Perdrix06,claudet2024covering} is one of them.

\begin{definition}Given $G=(V,E)$, a \emph{local set} $L$ is a non-empty subset of $V$ of the form $L = D \cup Odd_G(D)$ for some $D \se V$ called a \emph{generator}.
\end{definition}

Local sets of $G$ are precisely the supports of the Pauli operators $\mathcal P^G_D=(-1)^{|G[D]|}X_D Z_{Odd(D)}$, with $D$ a non-empty subset of qubits, that stabilises $\ket G$. Local sets are invariant under local complementation - hence their name: if $L$ is a local set in a graph $G$, so is in $G \star u$, but possibly with a distinct generator \cite{Perdrix06}. Local sets are the same for graphs that have the same cut-rank function \cite{claudet2024covering}, thus graphs corresponding to LU-equivalent graph states have the same local sets. 

A \emph{minimal local set} is a local set that is minimal by inclusion. Minimal local sets have either 1 or 3 generators, and in the latter case, the minimal local sets are of even size. This was proved originally in the stabilizer formalism in \cite{VandenNest05}, but an alternative graph-theoretic proof can be found in \cite{claudet2024covering}.

\begin{proposition}
    Given a minimal local set $L$, only two cases can occur:
    \begin{itemize}
        \item $L$ has exactly one generator,
        \item $L$ has exactly three (distinct) generators, of the form $D_0$, $D_1$, and $D_0 \Delta D_1$. This can only occur when $|L|$ is even.
    \end{itemize}
\label{prop:MLS2cases}
\end{proposition}

In the first case, the minimal local set $L$ is said to be of dimension 1; in the second, of dimension 2. The dimension of a minimal local set depends only on the cut-rank function, thus it is invariant by LU-equivalence. Just like (minimal) local sets, the dimension can be a useful tool to prove that two graph states are not LU-equivalent. For example, the two graphs of \cref{fig:k4} have the same minimal local sets, but they do not  have the same dimension. Indeed, $\{a,b,c,h\}$ is a minimal local set of dimension 2 in the first graph, while it is a minimal local set of dimension 1 in the second one, proving that the two graph states are not LU-equivalent.

Minimal local sets provide crucial information on local unitaries acting on graph states: it is known that if a local unitary $U$ transforms  $\ket G$ into $\ket {G'}$ and $L$ is a minimal local set of dimension $2$, then for any $v\in L$, $U_v$ must be a Clifford operator~\cite{VandenNest05}. Minimal local sets of dimension $1$ are more permissive but also provide some constraints on $U$:

\begin{lemma}\label{prop:comU}
Given two graphs $G$, $G'$ and a local unitary $U$ such that $\ket G = U\ket {G'}$, if $L$ is a $1$-dimensional minimal local set with generators respectively $D$ in $G$ and $D'$ in $G'$, then $$\mathcal P_D^G U= U\mathcal P_{D'}^{G'}$$
\end{lemma}

\begin{proof}
Following~\cite{VandenNest05,zeng2011transversality}, we consider the density matrix $\rho_G^L$ (resp. $\rho_{G'}^L$) obtained by tracing out the qubits outside $L$. According to \cite{Hein06,Hein04}:
$\rho^{L}_{G} = \frac{1}{2}(I +\mathcal P_{D}^{G} )$ and $\rho^{L}_{G'} = \frac{1}{2}(I +\mathcal P_{D'}^{G'})$. As $\rho^{L}_{G} = U\rho^L_{G'}U^\dagger$, we get $I +P_{D}^{G} = I +UP_{D'}^{G'} U^\dagger$, hence  $\mathcal P_{D}^{G}U =$ $ U\mathcal P_{D'}^{G'} $. 
\end{proof}

So, intuitively, every minimal local set $L$ induces some constraints on a local unitary transformation $U$ acting on the corresponding graph state, more precisely on the qubits of $U$ that are in $L$. It has been recently shown that minimal local sets cover every vertex of a graph \cite{claudet2024covering}, which unlocks the use of minimal local sets in characterizing local unitaries acting on graph states, as it guarantees that minimal local sets impose constraints on every qubit of the local unitary:

\begin{theorem}[\cite{claudet2024covering}]
    \label{thm:MLScover}
    Each vertex of any graph is contained in at least one minimal local set.
\end{theorem}

We abstract away the set of minimal local sets, which can be exponentially many in a graph (see  \cite{claudet2024covering}), as a simple labelling of the vertices that we call a \emph{type}.
\begin{definition} Given a graph $G$, a vertex $u$ is of type 
\begin{itemize}
    \item  X if for any generator $D$ of a minimal local set containing $u$, $u\in D \sm Odd(D)$,
    \item  Y  if for any generator $D$ of a minimal local set containing $u$, $u\in D \cap Odd(D)$, 
    \item Z  if for any generator $D$ of a minimal local set containing $u$, $u\in Odd(D) \sm D$, 
    \item $\bot$ otherwise. 
\end{itemize}
\end{definition}
The names X, Y and Z are chosen to match the Pauli operator at vertex $u$ in $\mathcal P^G_D$. Notice that a vertex involved in a minimal local set of dimension $2$ is necessarily of type $\bot$\footnote{If a vertex $v$ has the same type with respect to two distinct generators $D,D'$ of a minimal local set $L$, then $v$ would not be in the local set generated by $D\Delta D'$, contradicting the minimality of $L$.}.
We define $V^G_X \se V$ (resp. $V^G_Y$, $V^G_Z$, $V^G_\bot$) as the set of vertices of type X (resp. Y, Z, $\bot$) in $G$.  $V^G_X$, $V^G_Y$, $V^G_Z$ and $V^G_\bot$ form a partition of $V$: indeed, thanks to \cref{thm:MLScover}, every vertex has a type.

We show in the following a few properties of the vertex types. First notice that vertices of type Z represent at most half the vertices of a graph: 

\begin{lemma} \label{lemma:less_than_half_Z}
    For any graph $G$ of order $n$, $|V^G_Z| \ls \lfloor n/2\rfloor$.
\end{lemma}

\begin{proof}
    $V^G_Z$ contains no minimal local set. Indeed, as the generator of a local set is non-empty, at least one vertex of a given minimal local set $L$ is not of type Z. Conversely, any subset of the vertices of size $\lfloor n/2 \rfloor+1$ contains at least one minimal local set \cite{claudet2024covering}.
\end{proof}

Two LU-equivalent graph states share the same vertices of type $\bot$.

\begin{lemma}
    \label{lemma:LUbottom}
    Given two LU-equivalent graph states $\ket{G_1}$ and $\ket{G_2}$, $V^{G_1}_\bot = V^{G_2}_\bot$.
\end{lemma}

\begin{proof}
Let $U$ be a local unitary s.t. $\ket {G_1} = U\ket{G_2}$, and let $v\in V$ a vertex of type $\bot$ in $G_1$. If $v$ is in a minimal local set of dimension 2 in $G_1$, so is in $G_2$ as the dimension of minimal local sets is invariant under LU-equivalence. Otherwise, there exist two distinct minimal local sets $L$, $L'$, generated respectively by $D_1$ and by $D_1'$ in $G_1$ such that $\mathcal P_{D_1}^{G_1}$ and $\mathcal P_{D'_1}^{G_1}$ do not commute on qubit $v$.\footnote{$V=\bigotimes_u V_u$ commutes with $W = \bigotimes_u W_u$ on qubit $u_0$ if $V_{u_0}$ and $W_{u_0}$ commute.} Let $D_2$ (resp. $D_2'$) be the generator of $L$ (resp. $L'$) in $G_2$. According to \cref{prop:comU}, $\mathcal P_{D_1}^{G_1}= U\mathcal P_{D_2}^{G_2}U^\dagger$ and  $\mathcal P_{D'_1}^{G_1}= U\mathcal P_{D'_2}^{G_2}U^\dagger$, as a consequence, $\mathcal P_{D_2}^{G_2}$ and $\mathcal P_{D'_2}^{G_2}$ do not commute on qubit $v$, so $v$ must be of type $\bot$ in $G_2$. 
\end{proof}

A vertex having the same type in two LU-equivalent graph states implies strong constraints on the local unitaries that relate the two corresponding graph states.

\begin{lemma}
    \label{lemma:LUconstraints}
    If $\ket{G_1} =_{LU}\ket{G_2}$ i.e. $\ket {G_2} = U \ket {G_1}$, then $U=e^{i\phi}\bigotimes_{u\in V} U_u$ where:
    \begin{itemize}
        \item $U_u$ is a Clifford operator if  $u$ is of type $\bot$ in both $G_1$ and $G_2$,
        \item $U_u = X(\theta_u) Z^{b_u}$ if $u$ is of type X in both $G_1$ and $G_2$,
        \item  $U_u = Z(\theta_u) X^{b_u}$ if $u$ is of type Z in both $G_1$ and $G_2$.
    \end{itemize}
    with $b_u \in \{0,1\}$. Additionally, if $\ket{G_1}$ and $\ket{G_2}$ are LC$_r$-equivalent, there exists such a unitary $U$ where every angle satisfies $\theta_u = 0 \bmod \pi/2^r$.
\end{lemma}

\begin{proof}
    The first case, when $u$ is of type $\bot$, is a variant of the minimal support condition \cite{VandenNest04} and was proved in \cite{Rains97} for the particular case of minimal local sets of dimension 2. 
    If $u$ is of type X, let $L$ be a minimal local set such that $u \in L$. According to \cref{prop:comU}, $U_u X U_u^{\dagger} = e^{i \phi} X$. $U_u \ket +$ and $U_u \ket -$ are eigenvectors of $U_u X U_u^{\dagger}$ of eigenvalues respectively 1 and -1, implying $U_u X U_u^{\dagger} = \pm X$. If $U_u X U_u^{\dagger} = X$, then $U_u = e^{i \phi} X(\theta)$. If $U_u X U_u^{\dagger} = -X$, define ${U'} \defeq U_u Z$. ${U'} X {U'}^{\dagger} = X$, so ${U'} = e^{i \phi} X(\theta)$, thus $U_u = e^{i \phi} X(\theta) Z$. The proof is similar when $u$ is of type Z. If $\ket{G_1}$ and $\ket{G_2}$ are LC$_r$-equivalent, there exists $\ket {G_2} = U \ket {G_1}$ where $U=e^{i\phi}\bigotimes_{u\in V} U_u$, and each $U_u$ is in LC$_r$.
\end{proof}

To fully characterise the unitaries that relate two graph states using \cref{lemma:LUconstraints}, the type of every vertex needs to be the same. This is not the case in general, even for locally equivalent graphs\footnote{A simple example involves the complete graph $K_3$ on 3 vertices and the line graph $L_3$ on three vertices. $K_3$ and $L_3$ are related by a single local complementation, however every vertex of $K_3$ is of type Y, while $L_3$ contains one vertex of type Z and two vertices of type X.}. Thus, we introduce a standard form on graphs, such that two graphs in standard form corresponding to LU-equivalent graph states have the same types.

\subsection{Standard form}

We define a standard form of graphs up to local complementation. A graph  in standard form satisfies the following properties: all vertices are of type X, Z or $\bot$;  all the neighbours of a vertex of type X are of type Z (so in particular the vertices of type X form an independent set);  and  any vertex of type X is smaller than its neighbours according to the underlying total order $\prec$ of the vertices. In other words: 

\begin{definition} A graph $G$ is in said in standard form if 
    \begin{itemize}
        \item $V^G_Y = \emptyset$,
        \item $\forall u \in V^G_X$, any neighbour $v$ of $u$ is of type Z and satisfies $u\prec v$.
    \end{itemize} 
\end{definition}

Note that standard form is not unique in general, in the sense that a class of local equivalence of graphs may contain several graphs in standard form. For example, any graph with only vertices of type $\bot$ is in standard form, along with its entire orbit generated by local complementation. Conversely, each class of local equivalence contains at least a graph in standard form.

\begin{restatable}{proposition}{standardform}
    \label{prop:standardform}
    For any graph $G$, there exists a locally equivalent $G'$ in standard form.
\end{restatable}

%The proof of \cref{prop:standardform} is an algorithm that puts a graph in standard form by means of local complementation and is given in Appendix \ref{sec:proof:standardform}. 

\begin{proof}    
    To prove the proposition, we introduce an algorithm that transforms the input graph $G=(V,E)$ into a graph in standard form by means of local complementations. The action of the local complementation and the pivoting on the type of the vertices is given in the following table (the types of the unwritten vertices remain unchanged).
    \begin{center}
    \begin{tabular}{|c|c|}
    \hline
    \multicolumn{2}{|c|} {Type of $u$ in}\\
    $~~G~~$& $G\star u$\\
    \hline
    X&X\\
    \hline
    Y&Z\\
    \hline
    Z&Y\\
    \hline
    $\bot$&$\bot$\\
    \hline
    \end{tabular}\qquad\begin{tabular}{|c|c|}
    \hline
    \multicolumn{2}{|c|}{Type of $v {\in}N_G(u)$ in}\\
    $~~~G~~~$& $G\star u$\\
    \hline
    X&Y\\
    \hline
    Y&X\\
    \hline
    Z&Z\\
    \hline
    $\bot$&$\bot$\\
    \hline
    \end{tabular}
    \qquad \begin{tabular}{|c|c|}
    \hline
    \multicolumn{2}{|c|}{Type of $u$ (or $v$) in}\\
    $~~G~~$&$G\wedge uv$\\
    \hline
    X&Z\\
    \hline
    Y&Y\\
    \hline
    Z&X\\
    \hline
    $\bot$&$\bot$\\
    \hline
    \end{tabular}
    \end{center}

    The algorithm reads as follows:
    \begin{enumerate}
    \item If there is an XX-edge (i.e. an edge $uv$ such that both $u$ and $v$ are of type X): apply a pivoting on it. \\
    Repeat until there is no XX-edge left. 
    \item If there is an XY-edge: apply a local complementation of the vertex of type X, then go to step 1.
    \item If there is a vertex of type Y: apply a local complementation on it, then go to step 1.
    \item If there is an X$\bot$-edge: apply a pivoting on it. \\
    Repeat until there is no X$\bot$-edge left. 
    \item If there is an XZ-edge $uv$ such that $v\prec u$, apply a pivoting on $u v$.\\
    Repeat until for every XZ-edge $uv$, $u\prec v$.
    \end{enumerate}

    \noindent{\bf Correctness.~} When step 1 is completed, there is no XX-edge. Step 2 transforms the neighbours of type Y into vertices of type Z. No vertex of type Y is created as there is no XX-edge before the local complementation. When step 2 is completed, there is no XX-edge nor any XY-edge. Step 3 transforms the vertex of type Y into a vertex of type Z. No vertex of type Y is created as there is no XY-edge before the local complementation. When step 3 is completed, there is no vertex of type Y nor any XX-edge. In step 4, applying a pivoting on an X$\bot$-edge transforms the vertex of type X into a vertex of type Z. No XX-edge is created, as the vertex of type X has no neighbour of type X before the pivoting. When step 4 is completed, there is no vertex of type Y and each neighbour of a vertex of type X is of type Z. In step 5, applying a pivoting on an XZ-edge permutes the type of the two vertices, and preserves the fact that each neighbour of a vertex of type X is of type Z. When step 5 is completed, the graph is in standard form.
    
    {\bf Termination.} The quantity $2|V^G_Y|+|V^G_X|$ strictly decreases at each of the first 4 steps, which guarantees to reach step 5. At step 5, $V^G_X$ is updated as follows: exactly one vertex $u$ is removed from the set and is replaced by a vertex $v$ such that $v\prec u$, which guarantees the termination of step 5.
\end{proof}

Standard form implies a similar structure in terms of types assuming LU-equivalence.

\begin{restatable}{proposition}{sametypes}
    \label{lemma:sametypes}
    If $G_1$ and $G_2$ are both in standard form and $\ket{G_1} =_{LU}\ket{G_2}$, each vertex has the same type in $G_1$ and $G_2$, and any vertex $u$ of type X satisfies $N_{G_1}(u) = N_{G_2}(u)$.    
\end{restatable}

The proof of \cref{lemma:sametypes} is given in \cref{sec:proof:sametypes}.

\subsection{Graphical characterisation of local equivalence}

Thanks to the standard form, one can accommodate the types of two LU-equivalent graph states, to simplify the local unitaries mapping one  to the other:

\begin{restatable}{lemma}{standardformrotation}\label{lemma:standardform_implies_rotations}
    If $G_1$ and $G_2$ are both in standard form and $\ket{G_1} =_{LU}\ket{G_2}$, there exists $G'_1$locally equivalent to $G_1$ in standard form such that $\ket {G_2} = \bigotimes_{u\in V^{G_1}_X}X(\alpha_u)\bigotimes_{v\in V^{G_1}_Z} Z(\beta_v) \ket{G'_1}$.
\end{restatable}

The proof of \cref{lemma:standardform_implies_rotations} is given in Appendix \ref{sec:proof:standardform_implies_rotations}. Note that if $\ket{G_1}$ and $\ket{G_2}$ are LC$_r$-equivalent, we can choose the angles such that $\alpha_u, \beta_v = 0 \bmod \pi/2^r$. Additionally, $G_1$ and $G'_1$ are related by local complementation only on vertices of type $\bot$. They are strong constraints relating the angles of the X- and Z-rotations acting on different qubits:

\begin{restatable}{lemma}{conditionangles} \label{lemma:cond_angles}
Given $G_1$, $G_2$ in standard form, 
    if $\ket {G_2} = \bigotimes_{u\in  V_X^{G_1}}X(\alpha_u)\bigotimes_{v\in V_Z^{G_1}} Z(\beta_v) \ket{G_1}$,
    \begin{itemize}
        \item $\forall v \in V_Z^{G_1},~ \beta_v  = - \sum_{u \in N_{G_1}(v)\cap V_X^{G_1}}\alpha_u \bmod 2\pi$,
        \item $\forall k\in \mathbb N,\forall K \se V_Z^{G_1}$ of size $k+2$, $\sum_{u \in \Lambda_{G_1}^K \cap V_X^{G_1}}\alpha_u  = 0\bmod \dfrac{\pi}{2^{k+\delta(k)}}$.
    \end{itemize}
\end{restatable}

The proof of \cref{lemma:cond_angles} is given in Appendix \ref{sec:proof:cond_angles}. The constraints coincide with $r$-incidence, so, when angles are multiples of $\pi/2^r$, this unitary transformation implements an $r$-local complementation.

\begin{lemma} \label{lemma:implements_lc}
    Given $G_1$, $G_2$ in standard form, if $\ket {G_2} = \bigotimes_{u\in  V_X^{G_1}}X(\alpha_u)\bigotimes_{v\in V_Z^{G_1}} Z(\beta_v) \ket{G_1}$ and $\forall u \in V_X^{G_1}$, $\alpha_u= 0\bmod \pi/2^r$, then $\bigotimes_{u\in V_X^{G_1}}X(\alpha_u)\bigotimes_{v\in V_Z^{G_1}} Z(\beta_v)$ implements an $r$-local complementation over the vertices of $V_X^{G_1}$.
\end{lemma}

\begin{proof}
    Let us construct an $r$-incident independent multiset S such that $G_2 =  G_1 \star^r S$. We define $S$ on the vertices of $V_X^{G_1}$ such that $\forall u \in V_X^{G_1}$, $\alpha_u = \frac{S(u)\pi}{2^r}\bmod 2\pi$ with $S(u) \in [1,2^r)$. Note that $\sum_{u \in \Lambda_{G_1}^K \cap V_X^{G_1}}\alpha_u = \frac{\pi}{2^r}S\bullet \Lambda_G^K$. Hence, by \cref{lemma:cond_angles}, For any $k\in [0,r)$, and any $K\subseteq V\setminus S$ of size $k+2$, $S\bullet \Lambda_G^K$ is a multiple of $2^{r-k-\delta(k)}$ meaning that S is $r$-incident.    
    Also, by \cref{lemma:cond_angles}, $\beta_v = -\frac{\pi}{2^r} \sum_{u \in N_{G_1}(v)}S(u)\bmod 2\pi$. Thus, $\bigotimes_{u\in V_X^{G_1}}X(\alpha_u)\bigotimes_{v\in V_Z^{G_1}} Z(\beta_v)$ implements an $r$-local complementation on $S$.
\end{proof}

We can now easily relate LC$_r$-equivalence to $r$-local complementations for graphs in standard form:

\begin{lemma}
    \label{lemma:standardform_LCr_lc}
    If $G_1$ and $G_2$ are both in standard form and $\ket{G_1} =_{LC_{r}}\ket{G_2}$, then $G_1$ and $G_2$ are related by a sequence of local complementations on the vertices of type $\bot$ along with a single $r$-local complementation over the vertices of type X.
\end{lemma}

\begin{remark}
    The sequence of local complementations commutes with the $r$-local complementation, as vertices of type X and vertices of type $\bot$ do not share edges by definition of the standard form.
\end{remark}

\begin{proof}
    By \cref{lemma:standardform_implies_rotations}, there exists $G'_1$ locally equivalent to $G_1$ in standard form such that $\ket {G_2} = \bigotimes_{u\in V_X}X(\alpha_u)\bigotimes_{v\in V_Z} Z(\beta_v) \ket{G'_1}$ where $V_X$ (resp. $V_Z$) denotes the set of vertices of type X (resp. Z) in $G'_1$ and $G_2$, and $\forall u \in V$ of type X or Z, $\alpha_u, \beta_u = 0\bmod \pi/2^r$. By \cref{lemma:implements_lc}, $\bigotimes_{u\in V_X}X(\alpha_u)\bigotimes_{v\in V_Z} Z(\beta_v)$ implements an $r$-local complementation over the vertices of type X.
\end{proof}

Notice in particular that when there is no vertex of type $\bot$, a single $r$-local complementation is required.

\begin{corollary}\label{cor:constraint_r_lc}
    If  two $\bot$-free $r$-locally equivalent graphs $G_1$ and $G_2$ are both in standard form, they are related by a single $r$-local complementation on the vertices of type X.
\end{corollary}

We are ready to prove that $r$-local equivalence coincides with LC$_r$-equivalence.

\begin{theorem}\label{thm:LCr_lc}
    The following properties are equivalent:
    \begin{enumerate}
        \item $\ket{G_1}$ and $\ket{G_2}$ are LC$_r$-equivalent.
        \item $G_1$ and $G_2$ are $r$-locally equivalent.
        \item $G_1$ and $G_2$ are related by a sequence of local complementations along with a single $r$-local complementation.
    \end{enumerate}
\end{theorem}

\begin{proof}
    We proceed by cyclic proof.
    ($2\Rightarrow1$):~ Follows from \cref{prop:clif}. 
    ($3  \Rightarrow  2$):~ A local complementation is, in particular, an $r$-local complementation. 
    ($1  \Rightarrow  3$):~ Follows from \cref{prop:standardform} along with \cref{lemma:standardform_LCr_lc}.  
\end{proof}

More than just LC$_r$-equivalence, $r$-local complementations can actually characterise the LU-equivalence of graph states.

\begin{restatable}{theorem}{lulcr} \label{thm:LU_imply_LCr}
    If $\ket{G_1}$ and $\ket{G_2}$ are LU-equivalent then $G_1$ and $G_2$ are $(\lfloor n/2 \rfloor-1)$-locally equivalent, where $n$ is the order of the graphs.
\end{restatable}

The proof of \cref{thm:LU_imply_LCr} is given in \cref{sec:proof:LU_imply_LCr}. Remarkably, this result implies that for graph states, LU-equivalence reduces to equivalence up to operators in $\bigcup_{r \in \mathbb N} LC_r$, and even -- as it has been hinted in \cite{gross2007lu} -- to operators of the form $C_1 \bigotimes_{u\in V}Z(\alpha_u) C_2$ where $C_1, C_2$ are local Clifford operators and the angles are multiples of $\pi/2^r$ for some integer $r$. Notice that these local operators are actually those of the well-known Clifford hierarchy (see for example \cite{anderson2024groups}).

\subsection{Graph states whose class of LC and LU-equivalence coincide}

Since the LU-LC conjecture was disproven in \cite{Zeng07}, classes of graph states whose class of $LC_1$- and LU-equivalence coincide has been a subject of interest. We say that $LU \Leftrightarrow LC$ holds for a graph $G$ if $\ket{G} =_\text{LU} \ket{G'} \Leftrightarrow  \ket{G} =_{\text{LC}_1} \ket{G'}$ for any graph $G'$. It is known that $LU \Leftrightarrow LC$  holds for a graph $G$ if either one of the following is true: \textbf{(1)}~$G$ is of order at most 8 \cite{Hein06, CABELLO20092219}; \textbf{(2)}~$G$ is a complete graph \cite{VandenNest05}; \textbf{(3)}~every vertex of $G$ is of type $\bot$ (this is referred as the \textit{minimal support condition}) \cite{VandenNest05}; \textbf{(4)}~the graph obtained by removing all leafs (i.e. vertices of degree 1) vertices of $G$ has only vertices of type $\bot$ \cite{Zeng07};~\textbf{(5)} $G$ has no cycle of length 3 or 4 \cite{Zeng07};~\textbf{(6)} the stabilizer of $\ket G$ has rank less than 6 \cite{Ji07}. A review of these necessary conditions, especially for \textbf{(6)}, can be found in \cite{tzitrin2018local}. Our results imply a new criterion for $LU \Leftrightarrow LC$ based on the standard form introduced in the last section. This criterion is actually a sufficient and necessary condition, meaning it can be used to prove both that $LU \Leftrightarrow LC$ holds for some graph, or the converse.

\begin{restatable}{proposition}{lulc} \label{prop:lu-lc}
    Given a graph $G$, the following are equivalent: \begin{itemize}
        \item $LU \Leftrightarrow LC$ holds for $G$.
        \item For \textbf{some} graph locally equivalent to $G$ that is in standard form, any $r$-local complementation over the vertices of type X can be implemented by local complementations.
        \item For \textbf{any} graph locally equivalent to $G$ that is in standard form, any $r$-local complementation over the vertices of type X can be implemented by local complementations;
    \end{itemize}
\end{restatable}

The proof of \cref{prop:lu-lc} is given in \cref{app:proof:lu-lc}. Checking that $r$-local complementations can be implemented by local complementations is not easy in general, nonetheless it is convenient in many cases, for example when there are few vertices of type X or that they have low degree. Namely, this criterion is stronger than the minimal support condition, as graphs with only vertices of $\bot$ have no vertex of type X. Furthermore, it has been left as an open question in \cite{tzitrin2018local} whether $LU \Leftrightarrow LC$ holds for some instances of repeater graph states. We use our new criterion to easily prove that this is the case. For this purpose, we prove that $LU \Leftrightarrow LC$ holds for a broader class of graph.

\begin{proposition} \label{prop:lu-lc-degree1}
        $LU \Leftrightarrow LC$ holds for graphs where each vertex is either a leaf i.e.~a vertex of degree 1, or is adjacent to a leaf.        
\end{proposition}

\begin{proof}
    It is easy to prove (see \cref{app:leaf_standard_form}) that such a graph $G$ is in standard form (for some particular ordering of the vertices) and that the vertices of type X are exactly the leafs. Thus, any $r$-local complementation over the vertices of type X has no effect on $G$. According to \cref{prop:lu-lc}, this implies that $LU \Leftrightarrow LC$ holds for $G$.
\end{proof}

Such graphs include some instances of \emph{repeater graph states} \cite{azuma2023quantum}. A \emph{complete-graph-based repeater graph state} is a graph state whose corresponding graph of order $2n$ is composed of a complete graph of order $n$, along with $n$ leafs appended to each vertex.  A \emph{biclique-graph-based repeater graph state} is a graph state whose corresponding graph of order $4n$ is composed of a symmetric biclique (i.e. a symmetric bipartite complete graph) graph of order $2n$, along with $2n$ leafs appended to each vertex. Complete-graph-based repeater graph states are the all-photonic repeaters introduced in \cite{azuma2015all}. Biclique-graph-based repeater graph states are a variant introduced in \cite{russo2018photonic} that is more efficient in terms of number of edges.
It was left as an open question in \cite{tzitrin2018local} whether $LU \Leftrightarrow LC$ holds for complete-graph-based or biclique-graph-based repeater graph states (although it was proved for some variants). These graph states satisfy the condition of \cref{prop:lu-lc-degree1}, hence we can answer this question by the positive:

\begin{corollary}
    $LU \Leftrightarrow LC$ holds for complete-graph-based repeater graph states and biclique-graph-based repeater graph states.
\end{corollary}

%$LU \Leftrightarrow LC$ also holds for the graph states locally equivalent to the quantum repeaters proposed in \cite{ewert2017ultrafast}. Representatives of these graph state classes are very similar to the biclique-graph-based repeater graph states, but multiple leafs are appended to each vertex of the biclique.

\section{A hierarchy of generalised local equivalences}
\label{sec:hierarchy}

Two $r$-locally equivalent graphs are also $(r+1)$-locally equivalent. This can be seen graphically as a direct consequence of \cref{prop:monotonicity}, or using graph states through \cref{thm:LCr_lc}, as two LC$_{r}$-equivalent states are obviously LC$_{r+1}$-equivalent. The known counter examples to the LU-LC conjecture introduced in \cite{Ji07,Tsimakuridze17} are actually LC$_2$-equivalent graph states that are not LC$_1$-equivalent, thus the corresponding graphs are $2$-locally equivalent but not $1$-locally equivalent. There was however no known examples of graph states that are LC$_3$-equivalent but not LC$_2$-equivalent, and more generally no known examples of graph states that are LC$_r$-equivalent but not LC$_{r-1}$-equivalent for $r>2$. We introduce such examples in this section, by showing that for any $r$ there exist pairs of graphs that are $(r+1)$-locally equivalent but not $r$-locally equivalent leading to an infinite  hierarchy of generalised local equivalences.  Our proof is constructive: for any $r$ we exhibit pairs of graphs that are $(r+1)$-locally equivalent but not $r$-locally equivalent.

We introduce a family of bipartite graphs $C_{t,k}$ parametrised by two integers $t$ and $k$. This family is a variant of the (bipartite) Kneser graphs and uniform subset graphs. The graph $C_{t,k}$ is bipartite: the first independent set of vertices corresponds to the integers from $1$ to $t$, and the second independent set is composed of all the subsets of $[1,t]$ of size $k$. There is an edge between an integer $u\in [1,t]$ and a subset $A\subseteq [1,t]$ of size $k$ if and only if $u\in A$:

\begin{definition} For any $k\gs 1$ and $t\gs k$ two integers, let $C_{t,k} = (V,E)$ be a bipartite graph defined as: $V = [1,t] \cup \binom{[1,t]}{k}$~~and~~ $E = \{(u,A) \in [1,t]\times \binom{[1,t]}{k} ~|~ u\in A\}$. 
\end{definition}

\noindent

We introduce a second family of graphs $C'_{t,k}$, defined similarly to $C_{t,k}$, the independent set $[1,t]$ being replaced by a clique:

\begin{definition} For any $k\gs 1$ and $t\gs k$ two integers, let $C'_{t,k}=(V,E)$ be a graph defined as:  $V = [1,t] \cup \binom{[1,t]}{k}$ and $E = \{(u,A) \in [1,t]\times \binom{[1,t]}{k} ~|~ u\in A\} \cup \{(u,v)\in [1,t]\times [1,t]~|~u\neq v\}$.
\end{definition}

\begin{figure}[h]
\centering

\scalebox{0.7}{
\begin{tikzpicture}[xscale = 0.6,yscale=0.2]    
    \begin{scope}[every node/.style={circle,minimum size=18pt,thick,draw,fill=lipicsYellow, inner sep = 0pt}]
        \node (U1) at (-3,0) {$1$};
        \node (U2) at (-1,0) {$2$};
        \node (U3) at (1,0) {$3$};
        \node (U4) at (3,0) {$4$};
        \node (U12) at (-5,8) {\footnotesize $\{1,2\}$};
        \node (U13) at (-3,8) {\footnotesize$\{1,3\}$};
        \node (U14) at (-1,8) {\footnotesize$\{1,4\}$};
        \node (U23) at (1,8) {\footnotesize$\{2,3\}$};
        \node (U24) at (3,8) {\footnotesize$\{2,4\}$};
        \node (U34) at (5,8) {\footnotesize$\{3,4\}$};
    \end{scope}
    \begin{scope}[every node/.style={},
                    every edge/.style={draw=darkgray,very thick}]              
        \path [-] (U1) edge node {} (U12);
        \path [-] (U1) edge node {} (U13);
        \path [-] (U1) edge node {} (U14);
        \path [-] (U2) edge node {} (U12);
        \path [-] (U2) edge node {} (U23);
        \path [-] (U2) edge node {} (U24);
        \path [-] (U3) edge node {} (U13);
        \path [-] (U3) edge node {} (U23);
        \path [-] (U3) edge node {} (U34);
        \path [-] (U4) edge node {} (U14);
        \path [-] (U4) edge node {} (U24);
        \path [-] (U4) edge node {} (U34);
    \end{scope}
\end{tikzpicture}\qquad\qquad\raisebox{0cm}{
    \raisebox{-0.4cm}{
    \begin{tikzpicture}[xscale = 0.6,yscale=0.2]        
        \begin{scope}[every node/.style={circle,minimum size=18pt,thick,draw,fill=lipicsYellow, inner sep = 0pt}]
        \node (U1) at (-3,0) {$1$};
        \node (U2) at (-1,0) {$2$};
        \node (U3) at (1,0) {$3$};
        \node (U4) at (3,0) {$4$};
            \node (U12) at (-5,8) {\footnotesize$\{1,2\}$};
            \node (U13) at (-3,8) {\footnotesize$\{1,3\}$};
            \node (U14) at (-1,8) {\footnotesize$\{1,4\}$};
            \node (U23) at (1,8) {\footnotesize$\{2,3\}$};
            \node (U24) at (3,8) {\footnotesize$\{2,4\}$};
            \node (U34) at (5,8) {\footnotesize$\{3,4\}$};
        \end{scope}
        \begin{scope}[every node/.style={},
                        every edge/.style={draw=darkgray,very thick}]              
            \path [-] (U1) edge node {} (U12);
            \path [-] (U1) edge node {} (U13);
            \path [-] (U1) edge node {} (U14);
            \path [-] (U2) edge node {} (U12);
            \path [-] (U2) edge node {} (U23);
            \path [-] (U2) edge node {} (U24);
            \path [-] (U3) edge node {} (U13);
            \path [-] (U3) edge node {} (U23);
            \path [-] (U3) edge node {} (U34);
            \path [-] (U4) edge node {} (U14);
            \path [-] (U4) edge node {} (U24);
            \path [-] (U4) edge node {} (U34);

            \path [-] (U1) edge node {} (U2);
            \path [-] (U2) edge node {} (U3);
            \path [-] (U3) edge node {} (U4);
            \path [-] (U4) edge[bend left=75] node {} (U1);
            \path [-] (U1) edge[bend right=68] node {} (U3);
            \path [-] (U2) edge[bend right=68] node {} (U4);
        \end{scope}
\end{tikzpicture}}}
}

\caption{(Left) The graph $C_{4,2}$. (Right) The graph $C'_{4,2}$.}
\label{fig:bikneser}
\end{figure}
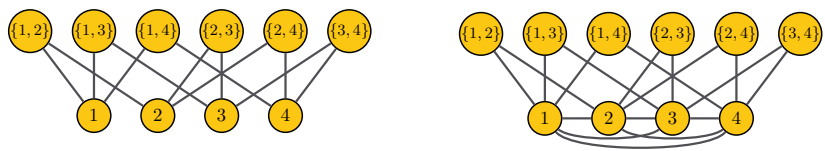
Examples of such graphs with parameters $t=4$ and $k=2$ are given in \cref{fig:bikneser}. For any odd $k \gs 3$, $t\gs k+2$, $C_{t,k}$ and $C'_{t,k}$ are in standard form for any ordering such that, for any $u \in \binom{[1,t]}{k}$ and for any $v \in [1,t]$, $u \prec v$. More precisely the vertices in $\binom{[1,t]}{k}$ are of type X and the vertices in $[1,t]$ are of type Z. The proof is given in \cref{app:c_standard_form}.

The following proposition gives a sufficient condition on $t$ and $k$ for the $r$-local equivalence of $C_{t,k}$ and $C'_{t,k}$.

\begin{proposition}\label{lemma:cond_lcr}
    For any $t\gs k\gs r+2$, $C_{t,k}$ and $C'_{t,k}$ are $r$-locally equivalent if
    \begin{equation*}
            \binom{t - 2}{k - 2}  = 2^{r-1}\bmod 2^r
          \text{~~~and~~~}  \forall i \in [ 1, r-1 ]\text{, ~~}\binom{t - i -2}{k - i -2}  = 0\bmod 2^{r-i}
    \end{equation*}
\end{proposition}

\begin{proof}
    $C_{t,k}$ and $C'_{t,k}$ are related by an $r$-local complementation on the set $\binom{[1,t]}{k}$ (which can be seen as the multiset where each vertex of $\binom{[1,t]}{k}$ appears once). Let $K \se [1,t]$ of size $k'+2$. $\binom{[1,t]}{k}\bullet \Lambda_G^K = \left|\left\{x \in \binom{[1,t]}{k}~|~K\se x\right\}\right| = \binom{t-k'-2}{k-k'-2}$ is a multiple of $2^{r-k'-\delta(k')}$ by hypothesis. Thus, $\binom{[1,t]}{k}$ is $r$-incident. Besides, for any $u,v \in [1,t]$, $\binom{[1,t]}{k} \bullet \Lambda_G^{u,v} = \binom{t-2}{k-2} = 2^{r-1}\bmod 2^{r}$ by hypothesis. Thus, $C_{t,k}\star^r\binom{[1,t]}{k} = C'_{t,k}$.
\end{proof}

The following proposition provides a sufficient condition on $t$ and $k$ for the non $r$-local equivalence of $C_{t,k}$ and $C'_{t,k}$.

\begin{proposition}
    \label{lemma:cond_not_lcr}
    For any odd $k \gs 3$, $t\gs k+2$, $C_{t,k}$ and $C'_{t,k}$ are not $r$-locally equivalent if $\binom{t}{2}$ is odd and $\binom{k}{2} = 0\bmod 2^{r}$.
\end{proposition}

\begin{proof}       
    Let us suppose by contradiction that $C_{t,k}$ and $C'_{t,k}$ are $r$-locally equivalent. 
    By \cref{cor:constraint_r_lc}, they are related by a single $r$-local complementation on 
    the vertices in $\binom{[1,t]}{k}$. Let $S$ be the multiset in $\binom{[1,t]}{k}$ such that $C_{t,k}\star^rS = C'_{t,k}$. For each $u,v \in [1,t]$, $S \bullet\Lambda_{C_{t,k}}^{u,v} = \sum_{x \in \Lambda_{C_{t,k}}^{u,v}}S(x) = 2^{r-1}\bmod 2^{r}$. Summing over all pairs $u,v \in [1,t]$, and, as by hypothesis $\binom{t}{2}$ is odd, $\sum_{u,v\in [1,t]} S \bullet\Lambda_{C_{t,k}}^{u,v} = \binom{k}{2}\sum_{x \in \binom{[1,t]}{k}}S(x)  = \binom{t}{2}2^{r-1}\bmod 2^r= 2^{r-1} \bmod 2^{r}$. The first part of the equation is true because $\sum_{u,v\in [1,t]} S \bullet\Lambda_{C_{t,k}}^{u,v} = \sum_{u,v\in [1,t]} \sum_{x \in \Lambda_{C_{t,k}}^{u,v}}S(x) = \sum_{x \in \binom{[1,t]}{k}} \sum_{u,v\in x} S(x) $. This leads to a contradiction as $\binom{k}{2} = 0 \bmod 2^{r}$ by hypothesis. 
\end{proof}

It remains to find, for any $r$, parameters $t$ and $k$ such that the corresponding graphs are $r$-locally equivalent but not $(r-1)$-locally equivalent. Fortunately, such parameters exist.

\begin{restatable}{theorem}{existenceCtk}
    For any $r \gs 2$, $C_{t,k}$ and $C'_{t,k}$ are $r$-locally equivalent but not $(r-1)$-locally equivalent when $k = 2^r +1$ and $t = 2^r + 2^{\lceil \log_2(r+2) \rceil} -1$. Thus, $\ket{C_{t,k}}$ and $\ket{C'_{t,k}}$ are LC$_{r}$-equivalent but not LC$_{r-1}$-equivalent.
\label{thm:existence_Ctk}
\end{restatable}

The proof of \cref{thm:existence_Ctk} is given in \cref{app:proof:existence_Ctk}.

\begin{remark}
    For the case $r=2$, we obtain the pair $(C_{7,5},C'_{7,5})$, which is the 28-vertex counter-example to the LU-LC conjecture from \cite{Tsimakuridze17}. In particular, this translates to an example of a pair of graph states that are LC$_2$-equivalent but not LC$_1$-equivalent. 
\end{remark}

Thus, while Ji et al. proved that there exist pairs of graph states that are LU-equivalent but not LC$_1$-equivalent \cite{Ji07}, we showed a finer result -- the existence of an infinite strict hierarchy of graph states equivalence between LC$_1$- and LU-equivalence. There exist LC$_2$-equivalent graph states that are not LC$_1$-equivalent, LC$_3$-equivalent graph states that are not LC$_2$-equivalent... On the other end, for any integer $r$, there exist LU-equivalent graph states that are not LC$_r$-equivalent.

\section{Conclusion}

In this paper, we have introduced of a graphical characterisation of LU-equivalence for graph states, and consequently for stabilizer states. To achieve this, we have introduced a generalisation of the local complementation, and leveraged the structures of minimal local sets in graphs.  
A key outcome of this characterisation is the establishment of a strict hierarchy of equivalences for graph states, which significantly advances our understanding of the gap between the LC- and LU-equivalence. 
Thanks to this graphical characterisation, we have also proven a conjecture regarding families of graphs states where LC- and LU-equivalence coincide. 

This graphical characterisation has additional potential applications, such as the search of minimal examples of graph states that are LU-equivalent but not LC-equivalent. The smallest known examples are of order 27, and it is known there is no counter example of order less than 8. More generally one can wonder whether there are smaller examples of graphs that are LC$_{r}$ equivalent but not LC$_{r-1}$ as the ones we have introduced are of order exponential en $r$. In terms of complexity, determining whether two graph states are LC-equivalent can be done in polynomial time. The graphical characterisation of LU-equivalence also offers new avenues for exploring the complexity of the LU-equivalence problem, as allows one to study this computational problem from a purely graphical standpoint.

%%% Garder pour la version arxiv/non-anonyme
\section*{Acknowledgements}
The authors want to thank David Cattaneo and Mehdi Mhalla for fruitful discussions on previous versions of this paper. This work is supported by the \emph{Plan
France 2030} through the PEPR integrated project EPiQ ANR-22-
PETQ-0007 and the HQI platform ANR-22-PNCQ-0002; and by the
European projects NEASQC and HPCQS.

%This work is supported by the PEPR integrated project EPiQ ANR-22-PETQ-0007 part of Plan France 2030, +HQI by the STIC-AmSud project Qapla’ 21-STIC-10,  and by the European projects NEASQC and HPCQS. (right ?) \Ncom{à repasser dessus.}

%%%%%%%%%%%%%%%%%%%%%%%%%%%%%%%%%%%%%%%%%%%%%%%%%%%%

\bibliography{reflipics}

\appendix

%\section{Proof of \cref{prop:stabilizer}}
\section{Proof of Proposition \ref{prop:stabilizer}}
\label{app:proof:stabilizer}

\stabilizer*

\begin{proof}
    We show that for any $D\subseteq V$, $(-1)^{|G[D]|}X_D Z_{Odd(D)}\ket G = \ket G$, by induction on the size of $D$. When $D=\emptyset$ the property is trivially true. Since $X$ and $Z$ anticommute, we have for any sets $A$ and $B$, $X_AZ_B = (-1)^{|A\cap B|}Z_BX_A$. As a consequence, for any non-empty $D$, we have, given $u\in D$, 
    \begin{equation*}
        \begin{split}
            (-1)^{|G[D]|}X_D Z_{Odd(D)}\ket G & = (-1)^{|G[D\setminus u]|}(-1)^{|N_G(u)\cap D\setminus u|}X_u X_{D\setminus u}Z_{N_G(u)}Z_{Odd(D\setminus u)}\ket G\\
            & =(-1)^{|G[D\setminus u]|}X_uZ_{N_G(u)}X_{D\setminus u}Z_{Odd(D\setminus u)}\ket G\\
            & = X_uZ_{N_G(u)}\ket G                   
        \end{split}
    \end{equation*}    
    by induction hypothesis. The fact that $\ket G$ is a fixed point of $X_uZ_{N_G(u)}$ terminates the proof.
\end{proof}

%\section{Proof of \cref{prop:LCLC}}
\section{Proof of Proposition \ref{prop:LCLC}}
\label{app:proof:LCLC}

\LCLC*

\begin{proof}
    First, we define the Z-weight $|C|_Z$ of an operator $C$, as the number of positions in $C$ that do not commute nor anti-commute with $Z$, i.e.~the number of $u\in V$, s.t. $\forall r\in\{-1,1\}$, $CZ_u\neq r Z_u C$. Notice that $C = e^{i\theta}\bigotimes_{u\in V} C_u$, where  every $1$-qubit Clifford  $C_u$ is  either of the form $PZ(\frac {k\pi}2)$, or $ PX(\pm\frac \pi 2)Z(\frac {k\pi}2)$, or $PHZ(\frac {k\pi}2)$, where $P$ is a Pauli operator and $k\in [0,3)$. Moreover, $|C_u|_Z=0$ if and only if $C_u=PZ(\frac {k\pi}2)$. 

    We prove, by induction on the Z-weight of $C$,  that $\ket {G_2}= C\ket {G_1}$ implies that there exists at most $\frac32|C|_Z$ local complementations that transforms $G_1$ into $G_2$ and that induces $C$ up to Pauli. 
    \\-- If $|C|_Z=0$, then $C$ must be a Pauli operator. Indeed, assume there exists $u$ s.t. $C_u=PZ(\pm\frac \pi2)$ where $P$ a Pauli. Then $\ket {G_2} = C\ket{G_1}=CX_uZ_{N_{G_1}(u)}\ket{G_1} = \mp iX_uZ_uZ_{N_{G_1}(u)}C\ket{G_1} =  \mp iX_uZ_uZ_{N_{G_1}(u)}\ket{G_2}$, leading to a contradiction as both $\ket {G_2}$ and $X_uZ_uZ_{N_{G_1}(u)}\ket{G_2}$ are in $\mathbb R^{2^n}$ where $n$ is the order of $G_2$. As consequence, $C$ is a Pauli operator, implying that $G_1= G_2$. 
    \\-- If there exists $u$ s.t. $C_u=PX(\pm\frac \pi 2)Z(\frac {k\pi}2)$, notice that $\ket{G_2\star u} = L_u^{G_2}\ket {G_2} =  C'\ket{G_1}$ with $C'=X_u(\frac \pi 2)Z_{N_{G_2}(u)}(-\frac \pi 2)C$. We have $|C'_u|_Z=0$ and for any $v\neq u$,  $|C'_v|_Z = |C_v|_Z$, so $|C'|_Z=|C|_Z-1$. By induction, there exists a sequence of (possibly repeating) vertices $a_1, \cdots, a_m$ such that $G_2\star u = G_1\star a_1 \star a_2 \star \cdots \star a_{m}$ and  the local Clifford operator that implements these $m$ local complementations is $PC'$ for some Pauli operator $P$. As a consequence $G_2= G_1\star a_1 \star a_2 \star \cdots \star a_{m}\star u$ and  the local Clifford operator that implements these $m+1$ local complementations is $X_u(\frac \pi 2)Z_{N_{G_2}(u)}(-\frac \pi 2)PC' =P'C$ for some Pauli $P'$,  moreover, we have $m+1\ls \frac32 |C'|_Z + 1 \ls \frac32 |C|_Z$. 
    \\-- Otherwise (i.e. $|C|_Z\neq 0$ and there is no $C_w$ of the form $PX(\pm\frac \pi 2)Z(\frac {k\pi}2)$), there exists $u$ s.t. $C_u= PHZ(\frac {k\pi}2)$. Notice that it implies that there exists also $v\in N_{G_2(u)}$ s.t.~$|C_v|_Z\neq 0$. Indeed, by contradiction, assume $\forall w\in N_{G_2}(u)$, $|C_w|_Z =0$, then $\ket {G_1} = C^\dagger \ket {G_2} = C^\dagger X_uZ_{N_{G_2}(u)}\ket {G_2} =  C^\dagger X_uZ_{N_{G_2}(u)}C\ket {G_1} = (-1)^a Z_uZ_{N_{G_2(u)}}\ket{G_1}$ for some $a \in \{0,1\}$. Thus, $\ket{G_1}= X_uZ_{N_{G_1(u)}}\ket{G_1} = (-1)^a X_uZ_{N_{G_1(u)}} Z_uZ_{N_{G_2(u)}}\ket{G_1} = - (-1)^a Z_uZ_{N_{G_2(u)}}\ket{G_1} = -\ket{G_1}$ leading to a contradiction. 
    Moreover, notice that $\ket{G_2\star u\star v\star u} = H_uH_vZ_{N_{G_2}(u)\cap N_{G_2}(v)}\ket {G_2} =  C'\ket{G_1}$ with $C'=H_uH_vZ_{N_{G_2}(u)\cap N_{G_2}(v)}C$. We have $|C'_u|_Z=|C'_v|_Z=0$ and for any other $w$, $|C'_w|_Z =|C_w|_Z$, so $|C'|_Z=|C|_Z-2$. By induction,  there exists a sequence of (possibly repeating) vertices $a_1, \cdots, a_m$ such that $G_2\star u\star v \star u = G_1\star a_1 \star a_2 \star \cdots \star a_{m}$ and  the local Clifford operator that implements these $m$ local complementations is $PC'$ for some Pauli operator $P$. As a consequence $G_2= G_1\star a_1 \star a_2 \star \cdots \star a_{m}\star u\star v\star u$ and  the local Clifford operator that implements these $m+3$ local complementations is $H_uH_vZ_{N_{G_2}(u)\cap N_{G_2}(v)}PC' =P'C$ for some Pauli $P'$,  moreover, we have $m+3\ls \frac32 |C'|_Z + 3 = \frac32 |C|_Z$. 
    ~\\It remains to prove \cref{rm:clifford}: The local Clifford $C'$ that implements the sequence of local complementations is equal to $PC$ for some Pauli $P$. Notice that $\ket{G_2} = PC \ket {G_1} = P\ket {G_2}$,
    thus according to \cref{prop:stabilizer}, there exists $D$ s.t.~$P= (-1)^{|G_2[D]|}X_DZ_{Odd(D)}$. 
    \end{proof}

\section{Basic properties of $r$-local complementation}
\label{app:prop_r_lc}

\rlcmonotonicity*

\begin{proof}
$2S$ is $(r+1)$-incident as for any $k\in [0,r+1)$ and any $K\subseteq V\setminus \supp(S)$ of size $k+2$, $2S\bullet \Lambda_G^K$ is a multiple of $2\times 2^{r-k-\delta(k)} = 2^{r+1-k-\delta(k)}$. Moreover, $S \bullet\Lambda_G^{u,v} = 2^{r-1}\bmod 2^{r}$ if and only if $2S \bullet\Lambda_G^{u,v} = 2^{r}\bmod 2^{r+1}$. $S$ is $(r-1)$-incident as for any $k\in [0,r-1)$ and any $K\subseteq V\setminus \supp(S)$ of size $k+2$, $S\bullet \Lambda_G^K$ is a multiple of $2^{r-k-\delta(k)}$, so is a multiple of $2^{r-1-k-\delta(k)}$. Moreover, for any vertices $u,v$, $S \bullet\Lambda_G^{u,v} = 2^{r-1}\bmod 2^{r}$ so $S \bullet\Lambda_G^{u,v} = 0\bmod 2^{r-1}$    
\end{proof}

\rlcsameneighborhood*

\begin{proof}
        If $u \in \supp(S)$, for any vertex $v$, $S \bullet\Lambda_G^{u,v} = 0$. Indeed, $N_G(u) \cap \supp(S) = \emptyset$ thus $\Lambda_G^{u,v} \cap \supp(S) = \emptyset$. 
\end{proof}

\rlcdisjointunion*

\begin{proof}
$S_1\sqcup S_2$ is $r$-incident as for any $k\in [0,r)$ and any $K\subseteq V\setminus \supp(S)$ of size $k+2$, $(S_1\sqcup S_2)\bullet \Lambda_G^K = S_1\bullet \Lambda_G^K + S_2\bullet \Lambda_G^K$ is a multiple of $2^{r-k-\delta(k)}$. $S_2$ is $r$-incident in $G\star^r S_1$ as for any set $K\subseteq V\setminus \supp(S)$ of size $k+2$, $S_2\bullet \Lambda_{G\star^r S_1}^K = S_2\bullet \Lambda_{G}^K$. Besides, for any vertices $u,v$, $S_1\sqcup S_2 \bullet\Lambda_G^{u,v} = S_1\bullet \Lambda_G^K + S_2\bullet \Lambda_G^K$ and 
\begin{equation*}
    \begin{split}
        u\sim_{(G\star^r S_1)\star^r S_2} v~ & \Leftrightarrow~\left(u\sim_{G\star^r S_1} v ~~\oplus~~ S_2\bullet\Lambda_{G\star^r S_1}^{u,v} = 2^{r-1}\bmod 2^{r}\right) ~\\
        & \Leftrightarrow~\left(u\sim_{G} v ~~\oplus~~ S_1\bullet\Lambda_{G}^{u,v} = 2^{r-1}\bmod 2^{r} ~~\oplus~~ S_2\bullet\Lambda_{G}^{u,v} = 2^{r-1}\bmod 2^{r}\right) \\
        & \Leftrightarrow~\left(u\sim_{G} v ~~\oplus~~ S_1\bullet\Lambda_{G}^{u,v} + S_2\bullet\Lambda_{G}^{u,v} = 2^{r-1}\bmod 2^{r}\right)        
    \end{split}
\end{equation*}
\end{proof}

%\section{Proof of \cref{prop:clif}}\label{sec:proof:clif}
\section{Proof of Proposition \ref{prop:clif}}\label{sec:proof:clif}

To prove  \cref{prop:clif} we use the following lemma :
\begin{lemma}\label{cor:i-e-oddbullet}
For any graph $G$,  any $K\subseteq V(G)$, and any multiset $S$ of vertices,
\begin{equation}\label{eq:ieodd}S\bullet Odd_G(K)= \sum_{R\subseteq K}(-2)^{|R|-1}S\bullet\Lambda_G^R\end{equation}
\end{lemma}

\begin{proof}
\cref{eq:ieodd} can be derived using the following inclusion exclusion identity for symmetric difference:  $\sum_{v\in \Delta_{u\in K} N_G(u)}\alpha_u = \sum_{R\subseteq K}(-2)^{|R|-1}\sum_{u\in \bigcap_{v\in R} N_G(v)}\alpha_u $. 
\end{proof}

\rlcwithclifford*

\begin{proof}
Let $\mathcal S := \supp (S)$ and $\bar {\mathcal S}:= V\setminus \mathcal S$. \\
Let $\ket \psi = \bigotimes_{u\in V}X\left(\frac {S(u)\pi}{2^r}\right)\bigotimes_{v\in V}Z\left(-\frac {\pi}{2^r}\sum_{u \in N_G(v)}S(u)\right)\ket{G}$. 
To show that $\ket{G\star^r S} = \ket \psi$, it is enough to show that for any $R\subseteq \mathcal S$, and any $L\subseteq \bar {\mathcal S}$, 
\begin{equation}\label{eq:basis}\bra+_{\mathcal S} Z_R \bra 0_{\bar {\mathcal S}} X_L \ket {G\star^r S} = \bra+_{\mathcal S} Z_R \bra 0_{\bar {\mathcal S}} X_L\ket \psi\end{equation}as $\{Z_R \ket +_{\mathcal S}\otimes X_L\ket0_{\bar {\mathcal S}}\}_{L,R}$ form a basis. Since $\bra 0 Z(\theta) = \bra 0$, $\bra + X(\theta)= \bra +$, $XZ(\theta)=e^{-i\theta} Z(-\theta)X$, and $ZX(\theta)=e^{-i\theta} X(-\theta)Z$, we get: 
\begin{align*}
\bra+_{\mathcal S} Z_R \bra 0_{\bar {\mathcal S}} X_L\ket {\psi}&= \bra+_{\mathcal S} Z_R \bra 0_{\bar {\mathcal S}} X_L\bigotimes_{u\in V}X\left(\frac {S(u)\pi}{2^r}\right)\bigotimes_{v\in V}Z\left(-\frac {\pi}{2^r}\sum_{u \in N_G(v)}S(u)\right)\ket{G}\\
&= e^{\frac{i\pi}{2^r}\left( -\sum_{u\in R} S(u)+ \sum_{v\in L}\sum_{u\in N_G(v)}S(u)\right)}  \bra+_{\mathcal S} Z_R \bra 0_{\bar {\mathcal S}} X_L\ket{G}\\
&= e^{\frac{i\pi}{2^r} \left(-S\bullet R +\sum_{v\in L}  S\bullet N_G(v)\right)} \bra+_{\mathcal S} Z_R \bra 0_{\bar {\mathcal S}} X_L\ket{G}
\end{align*}
Moreover, 
\begin{align*}
\bra+_{\mathcal S} Z_R \bra 0_{\bar {\mathcal S}} X_L \ket {G}&= (-1)^{|G[L]|}\bra+_{\mathcal S} Z_R\bra 0_{\bar {\mathcal S}} Z_{Odd_G(L)} \ket {G}\\
&=\frac{(-1)^{|G[L]|}}{\sqrt{2^{|\bar{\mathcal S}|}}} \bra +_{\mathcal S}Z_{R\Delta (Odd_G(L)\cap \mathcal S)} \ket {G[\mathcal S]}\\
&=\frac{(-1)^{|G[L]|}}{\sqrt{2^{|\bar{\mathcal S}|}}} \bra +_{\mathcal S}Z_{R\Delta (Odd_G(L)\cap \mathcal S)} \ket {+}_{\mathcal S}\\
&=\begin{cases}\frac{(-1)^{|G[L]|}}{\sqrt{2^{|\bar{\mathcal S}|}}} & \text{when $R = Odd_G(L)\cap \mathcal S$}\\0&\text{otherwise}
\end{cases}
\end{align*}
Moreover, thanks to \cref{loccompneighS}, for any $L\subseteq \mathcal S$, $Odd_G(L) \cap \mathcal S = Odd_{G\star^r S}(L) \cap \mathcal S$. So, \cref{eq:basis} is trivially satisfied when $R\neq Odd_G(L)\cap \mathcal S$. When $R= Odd_G(L)\cap \mathcal S$, notice also  that $ \sum_{u\in Odd_G(L)\cap \mathcal S} S(u) =  \sum_{u\in Odd_G(L)} S(u)$ as $\mathcal S$ is, by definition, the support of $S$. 
Thus, it is enough to show that for any $L\subseteq \bar{\mathcal S}$, \begin{equation*}\label{eq:LCR} S\bullet Odd_G(L) = 2^r(|G[L]|+ |(G\star^rS)[L]|) +  \sum_{v\in L} S\bullet N_G(v)  \bmod 2^{r+1}\end{equation*}

Thanks to \cref{cor:i-e-oddbullet}, we have: $S\bullet Odd_G(L) = \sum_{T\subseteq L} (-2)^{|T|-1}S\bullet\Lambda^T_G$. The $r$-incidence condition implies that $(-2)^{|T|-1}S\bullet\Lambda^T_G= 0\bmod 2^{r+1}$ when $|T|>2$. The $r$-incidence condition also  implies that $S\bullet\Lambda^T_G $ is a multiple of $2^{r-1}$ when $|T|=2$, so $S\bullet\Lambda^T_G  = d \bmod 2^{r}$ with $d\in \{0, 2^{r-1}\}$. By definition of the $r$-local complementation, $d\neq 0$ iff the edge between the two vertices is toggled, so $\sum_{T\subseteq L, |T|=2} S\bullet\Lambda^T_G$ is equal to $2^{r-1}$ times the numbers of edges (which both end points are in $L$) toggled by the $r$-local complementation modulo $2^r$, i.e. $2^{r-1} (|G[L]|+ |(G\star^rS)[L]|) \bmod 2^r$. Thus, 
\begin{align*}
S\bullet Odd_G(L) &= \sum_{T\subseteq L} (-2)^{|T|-1}S\bullet\Lambda^T_G\\
&= \sum_{T\subseteq L, |T| = 2} (-2)^{|T|-1}S\bullet\Lambda^T_G +  \sum_{T\subseteq L, |T| = 1} (-2)^{|T|-1}S\bullet\Lambda^T_G  \bmod 2^{r+1}  \\
&= \sum_{T\subseteq L, |T| = 2} (-2) S\bullet\Lambda^T_G +  \sum_{T\subseteq L, |T| = 1}S\bullet\Lambda^T_G \bmod 2^{r+1} \\
&= -2^{r} (|G[L]|+ |(G\star^rS)[L]|) +  \sum_{v\in  L}S\bullet N_G(v) \bmod 2^{r+1} \\
&= 2^{r} (|G[L]|+ |(G\star^rS)[L]|) +  \sum_{v\in  L}S\bullet N_G(v) \bmod 2^{r+1}
\end{align*}
\end{proof}

\section{Proof of Proposition \ref{lemma:sametypes}}
\label{sec:proof:sametypes}

We begin by proving the following lemma.

\begin{lemma} \label{lemma:closed_neighborhood_is_MLS}
    If $G$ is in standard form, for any $u \in V^G_X$%\Scom{of type X (?)}
    , $\{u\} \cup N_G(u)$ is a minimal local set.
\end{lemma}

\begin{proof}
    A minimal local set $L = D \cup Odd_G(D)$ in $\{u\} \cup N_G(u)$ is such that $D=\{u\}$, as $D$ is non-empty and each vertex of $N_G(u)$ is of type Z.
\end{proof}

\sametypes*

\begin{proof}
First, by \cref{lemma:LUbottom}, $V^{G_1}_\bot = V^{G_2}_\bot$.
Let us show that $V^{G_1}_X = V^{G_2}_X$. Let $u$ be a vertex of type X in $G_1$ and suppose by contradiction that $u$ is of type Z in $G_2$. By \cref{lemma:closed_neighborhood_is_MLS}, $L \defeq \{u\} \cup N_{G_1}(u)$ is a minimal local set and $u$ is the smallest element in $L$ (in both $G_1$ and $G_2$). Moreover, $L = D \cup Odd_{G_2}(D)$ and $u \in Odd_{G_2}(D) \sm D$ by definition of type Z. Thus, $u$ is adjacent to at least one vertex $v \in D$ in $G_2$. $v$ is of type X in $G_2$ as, by definition of the standard form, there is no vertex of type Y, and $N_{G_1}(u)$ contains no vertex of type $\bot$. This is a contradiction with the definition of the standard form, as $u \prec v$. By symmetry, any vertex of type X in $G_2$ is also of type X in $G_1$. Overall, each vertex has the same type in $G_1$ and $G_2$.
By \cref{lemma:closed_neighborhood_is_MLS}, for any vertex $u$ of type X, $L \defeq \{u\} \cup N_{G_1}(u)$ is a minimal local set, thus there exists $D\subseteq L$ such that $L = D\cup Odd_{G_2}(D)$. As $D$ is non-empty and does not contain any vertex of type Z, $D=\{u\}$, implying that $N_{G_1}(u) = N_{G_2}(u)$.
\end{proof}

%\section{Proof of \cref{lemma:standardform_implies_rotations}}
\section{Proof of Lemma \ref{lemma:standardform_implies_rotations}}
\label{sec:proof:standardform_implies_rotations}

\standardformrotation*

\begin{proof}
    By \cref{lemma:sametypes}, each vertex has the same type in $G_1$ and $G_2$, so, by \cref{lemma:LUconstraints} there exists $U$ such that $\ket {G_2} = U \ket {G_1}$, with $U=e^{i\phi}\bigotimes_{u\in V} U_u$ where:
    \begin{itemize}
        \item $C_u \defeq U_u$ is a Clifford operator if  $u$ is of type $\bot$,
        \item $U_u = X(\theta_u) Z^{b_u}$ if $u$ is of type X,
        \item  $U_u = Z(\theta_u) X^{b_u}$ if $u$ is of type Z.
    \end{itemize}
    with $b_u \in \{0,1\}$. Additionally, if $\ket{G_1}$ and $\ket{G_2}$ are LC$_r$-equivalent, we choose $U$ so that the angles satisfy $\theta_u = 0\bmod \pi/2^r$. Define $V_X \defeq V^{G_1}_X = V^{G_2}_X$, $V_Z \defeq V^{G_1}_Z = V^{G_2}_Z$, and $V_\bot \defeq V^{G_1}_\bot = V^{G_2}_\bot$. Summing up,

    $$ \ket {G_2} = e^{i \phi} \bigotimes_{u\in V_X}X(\alpha_u) Z^{b_u} \bigotimes_{u\in V_Z} Z(\beta_u) X^{b_u} \bigotimes_{u\in \bot}C_u \ket{G_1} $$
    As for any vertex $u \in V$, $X_u Z_{N_{G_1}(u)} \ket{G_1} = \ket{G_1}$, for possibly different angles $\beta_u$ (albeit equal modulo $\pi/2$), values $b_u$ and Clifford operators $C_u$, $$ \ket {G_2} = e^{i \phi} \bigotimes_{u\in V_X}X(\alpha_u) Z^{b_u} \bigotimes_{u\in V_Z} Z(\beta_u) \bigotimes_{u\in \bot}C_u \ket{G_1} $$

    By applying $\bra{+}_{V_X}\bra{0}_{V_Z}$ on both sides of the previous equation we get, on the LHS, $\bra{+}_{ V_X}\bra{0}_{V_Z} \ket {G_2} = \frac1{\sqrt{2^{|V_Z|}}} \bra{+}_{V_X} \ket{+}_{V_X} \ket{G_2[\bot]} = \frac1{\sqrt{2^{|V_Z|}}} \ket{G_2[\bot]}$ as $V_X$ is an independent set in $G_1$; and, on the RHS, $\begin{cases}\frac{e^{i\phi}}{\sqrt{2^{|V_Z|}}}\otimes_{u\in\bot} C_u \ket{G_1[\bot]} & \text{if $\forall u\in V_X, b_u=0$}\\0&\text{otherwise}%\tag{A}
    \end{cases}$, since $\bra 0 Z(\theta)=\bra 0$, $\bra +X(\theta) = \bra +$ and $\bra +Z = \bra -$. Hence, every $b_u = 0$, and $\ket {G_2[\bot]} = e^{i \phi} \bigotimes_{u\in \bot}C_u \ket{G_1[\bot]}$ (if $\bot = \emptyset$ then $e^{i\phi} = 1$).

    Using \cref{prop:LCLC}, there exist a sequence of (possibly repeating) vertices $a_1, \cdots, a_m \in \bot$ and a set $D \se \bot$ such that $G_2[\bot] = G_1[\bot] \star a_1 \star a_2 \star \cdots \star a_m$ and $e^{i \phi}\bigotimes_{u\in \bot}C_u = L_{a_m} \cdots L_{a_2} L_{a_1}$, where $L_{a_i} = L_{a_i}^{{G_1}[\bot] \star a_1 \star \cdots \star a_{i-1}}$.

    These operators do not implement a local complementation on the full graph. For any $i \in [1, m]$, let us introduce $$L'_{a_i} = L_{a_i} Z\left(- \frac{\pi}{2}\right)_{N_{G_1 \star a_1 \star \cdots \star a_{i-1}}(a_i)\sm \bot} = X\left(\frac{\pi}{2}\right)_{a_i} Z\left(-\frac{\pi}{2}\right)_{N_{G_1 \star a_1 \star \cdots \star a_{i-1}}(a_i)} = L_{a_i}^{G_1 \star a_1 \star \cdots \star a_{i-1}}$$

    Then, for possibly different angles $\beta_u$ (albeit equal modulo $\pi/2$),
    \begin{equation*}
        \begin{split}
            \ket {G_2} &= \bigotimes_{u\in V_X}X(\alpha_u)  \bigotimes_{u\in V_Z} Z(\beta_u) L'_{a_m} \cdots L'_{a_2} L'_{a_1} \ket{G_1}\\
            & = \bigotimes_{u\in V_X}X(\alpha_u)  \bigotimes_{u\in V_Z} Z(\beta_u) \ket{G_1 \star a_1 \star \cdots \star a_m }
        \end{split}
    \end{equation*}

\end{proof}

%\section{Proof of \cref{lemma:cond_angles}}
\section{Proof of Lemma \ref{lemma:cond_angles}}
\label{sec:proof:cond_angles}

We begin by proving the following lemma.

\begin{lemma} \label{lem:cond_angles}
    Given $G_1$, $G_2$ in standard form, if $\ket {G_2} = \bigotimes_{u\in V^{G_1}_X}X(\alpha_u)\bigotimes_{v\in V^{G_1}_Z} Z(\beta_v) \ket{G_1}$, then for any $K \se V^{G_1}_Z$, $\sum_{v\in K}\beta_{v}+\sum_{u\in Odd_{G_1}(K)\cap V^{G_1}_X} \alpha_u=|G_1\Delta G_2[K]|\pi  \bmod 2\pi$. 
\end{lemma}
    
\begin{proof} For compactness, we write $V_X \defeq V^{G_1}_X = V^{G_2}_X$ and $V_Z \defeq V^{G_1}_Z = V^{G_2}_Z$.
    According to  \cref{lemma:sametypes}, for any $u\in V_X$, $N_{G_1}(u) = N_{G_2}(u)$, so for any $ v\in V \setminus V_X$, $N_{G_1}(v)\cap V_X= N_{G_2}(v)\cap V_X$. 
    Notice that for any $\theta$, $X(\theta)Z = e^{-i\theta}ZX(-\theta)$ and $Z(\theta)X= e^{-i\theta}XZ(-\theta)$ so, for any $K\subseteq V_Z$, 
    \begin{align*} \ket{G_2} &= (-1)^{|G_2[K]|}X_{K}Z_{Odd_{G_2}(K)}\ket{G_2}\\
    &=(-1)^{|G_2[K]|} X_{K}Z_{Odd_{G_2}(K)}\bigotimes_{u\in V_X}X(\alpha_u)\bigotimes_{v\in V_Z} Z(\beta_v) \ket{G_1}\\
    &=(-1)^{|G_2[K]|}X_{K}Z_{Odd_{G_2}(K)}\bigotimes_{u\in V_X}X(\alpha_u)\bigotimes_{v\in V_Z} Z(\beta_v) (-1)^{|G_1[K]|}X_{K}Z_{Odd_{G_1}(K)}\ket{G_1}\\
    &=(-1)^{|G_2\Delta G_1[K]|}e^{-i\left(\sum_{v\in K}\beta_{v}+\sum_{u\in Odd_{G_1}(K)\cap V_X} \alpha_u\right)}Z_{Odd_{G_1}(K)\Delta Odd_{G_2}(K)}\\ & \hspace{227pt} \bigotimes_{u\in V_X}X(\pm \alpha_u)\bigotimes_{v\in V_Z} Z(\pm \beta_v) \ket{G_1}
    \end{align*}
    By applying $\bra{+}_{V_X}\bra{0}_{V\setminus V_X}$ on both sides of the previous equation we get, on the LHS,  $\bra{+}_{V_X}\bra{0}_{V\setminus V_X}\ket{G_2} = \frac1{\sqrt{2^{n-|V_X|}}}\bra{+}_{V_X}\ket{G_2[V_X]} = \frac1{\sqrt{2^{n-|V_X|}}}$, as $V_X$ is an independent set in $G_2$; and, on the RHS, for any $K\subseteq V_Z$,
    \begin{align*}
        & \bra{+}_{V_X}\bra{0}_{V\setminus V_X} (-1)^{|G_2\Delta G_1[K]|} e^{-i\left(\sum_{v\in K}\beta_{v}+\sum_{u\in Odd_{G_1}(K)\cap V_X} \alpha_u\right)}Z_{Odd_{G_1}(K)\Delta Odd_{G_2}(K)} \\ & \hspace{243pt} \bigotimes_{u\in V_X}X(\pm \alpha_u)\bigotimes_{v\in V_Z} Z(\pm \beta_v) \ket{G_1}\\ & = \frac{e^{i\left(|G_1\Delta G_2[K]|\pi-\sum_{v\in K}\beta_{v}-\sum_{u\in Odd_{G_1}(K)\cap V_X} \alpha_u \right)}}{\sqrt{2^{n-|V_X|}}}
    \end{align*}
    As a consequence, $\sum_{v\in K}\beta_{v}+\sum_{u\in Odd_{G_1}(K)\cap V_X} \alpha_u=|G_1\Delta G_2[K]|\pi  \bmod 2\pi$.
\end{proof}

\conditionangles*

\begin{proof} For compactness, we write $V_X \defeq V^{G_1}_X = V^{G_2}_X$ and $V_Z \defeq V^{G_1}_Z = V^{G_2}_Z$. According to \cref{lem:cond_angles}, for any $K \se V_Z$,  $\sum_{v\in K}\beta_{v}+\sum_{u\in Odd_{G_1}(K)\cap V_X} \alpha_u=|G_1\Delta G_2[K]|\pi  \bmod 2\pi$, thus in particular,  when $K=\{v\}$, $\beta_{v}  = - \sum_{u \in N_{G_1}(v)\cap V_X}\alpha_u \bmod 2\pi$. 
We can then show, by induction on the size of $K$, that $\sum_{u \in \Lambda_{G_1}^K \cap V_X}\alpha_u  = 0\bmod \dfrac{\pi}{2^{k+\delta(k)}}$, where $k=|K|-2$.
For $k=0$, since $\Lambda_{G_1}^{\{v_0,v_1\}} = N_{G_1}(v_0)\cap N_{G_1}(v_1)$, we have 
\begin{align*}
2\sum_{u \in \Lambda_{G_1}^{\{v_0,v_1\}} \cap V_X}\alpha_u &= \sum_{u \in N_{G_1}(v_0) \cap V_X}\alpha_u  + \sum_{u \in N_{G_1}(v_1) \cap V_X}\alpha_u - \sum_{u \in Odd_{G_1}(\{v_0,v_1\}) \cap V_X}\alpha_u\\
&= -\beta_{v_0} - \beta_{v_1} + \beta_{v_0} + \beta_{v_1} + |G_1\Delta G_2 [\{v_0,v_1\}]|\pi \bmod 2\pi\\
&= |G_1\Delta G_2 [\{v_0,v_1\}]|\pi \bmod 2\pi
\end{align*}

So $\sum_{u \in \Lambda_{G_1}^{\{v_0,v_1\}} \cap V_X}\alpha_u  = 0\bmod{\frac \pi 2}$. 

To prove the general case we use the following including-excluding-like identity for symmetric differences and intersections: given a set $V$ and a family $(A_v)_{v\in K}$ s.t. for all $v\in K$, $A_v\subseteq V$, for any $(\alpha_u)_{u\in V}$,  $\sum_{u\in \Delta_{v\in K}A_v}\alpha_u =\sum_{R\subseteq K}(-2)^{|R|-1}\sum_{u\in \bigcap_{v\in R} A_v}\alpha_u$. 
 
Thus, for any $K \se V_Z$, 
\begin{align*}
& \sum_{u\in Odd_{G_1}(K)\cap V_X} \!\!\!\!\!\alpha_u\\ &= \sum_{R\subseteq K}(-2)^{|R|-1}\sum_{u\in \bigcap_{v\in \Lambda^R_{G_1}\cap V_X}}\!\!\!\!\!\alpha_u \\
&= \sum_{R\subseteq K, |R|>2}(-2)^{|R|-1}\!\!\!\!\!\!\!\!\!\sum_{u\in \bigcap_{v\in \Lambda^R_{G_1}\cap V_X}}
\!\!\!\!\!\alpha_u  -2 \sum_{v_0,v_1\in K} \sum_{u\in \bigcap_{v\in \Lambda^{\{v_0,v_1\}}_{G_1}\cap V_X}}
\!\!\!\!\!\alpha_u
+ \sum_{v_0\in K} \sum_{u\in N_{G_1}(v_0)\cap V_X}\!\!\!\!\!\alpha_u\\
&= \sum_{R\subseteq K, |R|>2}(-2)^{|R|-1}\!\!\!\!\!\!\!\!\!\sum_{u\in \bigcap_{v\in \Lambda^R_{G_1}\cap V_X}}\!\!\!\!\!\alpha_u + |G_1\Delta G_2[K]|\pi
- \sum_{v_0\in K} \beta_{v_0} \bmod 2\pi
 \end{align*}

As a consequence, for any $K \se V_Z$,
  $ \sum_{R\subseteq K, |R|>2}(-2)^{|R|-1}\sum_{u\in \bigcap_{v\in \Lambda^R_{G_1}\cap V_X}}\alpha_u =0\bmod 2\pi$, as $\sum_{v\in K}\beta_{v}+\sum_{u\in Odd_{G_1}(K)\cap V_X} \alpha_u=|G_1\Delta G_2[K]|\pi  \bmod 2\pi$.  So by induction on the size of $K$, for any $K$ s.t. $|K|>2$, $2^{|K|-1}\sum_{u \in \Lambda_{G_1}^K \cap V_X}\alpha_u =0\bmod 2\pi$, as 
\begin{align*}& \sum_{R\subseteq K, |R|>2}(-2)^{|R|-1}\sum_{u\in \bigcap_{v\in \Lambda^R_{G_1}\cap V_X}}\!\!\!\!\!\alpha_u \\
&= 2^{|K|-1}\sum_{u\in \bigcap_{v\in \Lambda^K_{G_1}\cap V_X}}\!\!\!\!\!\alpha_u + \sum_{R\subsetneq K, |R|>2}(-2)^{|R|-1}\sum_{u\in \bigcap_{v\in \Lambda^R_{G_1}\cap V_X}}\!\!\!\!\!\alpha_u\\  &= 2^{|K|-1}\sum_{u\in \bigcap_{v\in \Lambda^K_{G_1}\cap V_X}}\!\!\!\!\!\alpha_u + 0 \bmod 2\pi \end{align*}

\end{proof}

%\section{Proof of \cref{thm:LU_imply_LCr}}
\section{Proof of Theorem \ref{thm:LU_imply_LCr}}
\label{sec:proof:LU_imply_LCr}

\lulcr*

\begin{proof}    
    Suppose $\ket{G_1}$ and $\ket{G_2}$ LU-equivalent. By \cref{prop:standardform} along with \cref{lemma:standardform_implies_rotations}, there exist $G'_1$ locally equivalent to $G_1$ and $G'_2$ locally equivalent to $G_2$ both in standard form such that $\ket {G'_2} = \bigotimes_{u\in V_X}X(\alpha_u)\bigotimes_{v\in V_Z} Z(\beta_v) \ket{G'_1}$ where $V_X$ (resp. $V_Z$) denotes the set of vertices of type X (resp. Z) in $G'_1$ and $G'_2$.

    If two vertices $u$ and $v$ are twins, i.e.~$N_G(u)=N_G(v)$, then $X_u(\theta_1)X_v(\theta_2)\ket G = X_u(\theta_1 + \theta_2) \ket G$. Indeed, $X_u(\theta)\ket G = e^{i\frac\theta 2}(\cos(\frac \theta 2)\ket G + i\sin(\frac \theta 2)X_u\ket G) = e^{i\frac\theta 2}(\cos(\frac \theta 2)\ket G + i\sin(\frac \theta 2)Z_{N_G(u)}\ket G) = X_v(\theta)\ket G $. Hence, without loss of generality, we suppose that if $u,v \in V_X$ are twins, then at most one of the angles $\alpha_u, \alpha_v$ is non-zero.

    Moreover, if a vertex $u$ is a leaf, i.e.~there exists a unique $v \in V$ such that $u \sim_{G} v$, then $X_u(\theta) Z_v(-\theta) \ket G = \ket G$. Indeed, $X_u(\theta) \ket G = e^{i\frac\theta 2}(\cos(\frac \theta 2)\ket G + i\sin(\frac \theta 2)Z_{v}\ket G) = Z_v(\theta) \ket G$. Hence, without loss of generality, we suppose that if $u \in V_X$ is a leaf, then $\theta = 0$. 
    
    If a vertex $u$ is an isolated vertex, $X_u(\theta) \ket G = \ket G$. Hence, without loss of generality, we suppose that if a vertex $u \in V_X$ is an isolated vertex, then $\alpha_u = 0$.

    Let us show that for every $u \in V_X$, $\alpha_u=0\bmod \pi/2^{|V_Z|-2+\delta(|V_Z|-2)}$ where $\delta$ is the Kronecker delta. If $|V_Z|\ls 1$, the proof is trivial as every vertex in $V_X$ is a leaf or an isolated vertex. Else, we prove by induction over the size of a set $K \se V_Z$ that for every $u\in \Lambda_{G'_1}^{V_Z\setminus K} \cap V_X$, $\alpha_u = 0\bmod \pi/2^{|V_Z|-2+\delta(|V_Z|-2)}$.

    \begin{itemize}
        \item If $K = \emptyset$: there is at most one vertex $u$ in $\Lambda_{G'_1}^{V_Z\setminus K} \cap V_X = \Lambda_{G'_1}^{V_Z} \cap V_X$ such that $\alpha_u$ is non-zero. By \cref{lemma:cond_angles}, $\alpha_u = 0\bmod \pi/2^{|V_Z|-2+\delta(|V_Z|-2)}$.
        \item Else, there is at most one vertex $u$ in $\left(\Lambda_{G'_1}^{V_Z\setminus K} \sm \bigcup_{K' \varsubsetneq K} \Lambda_{G'_1}^{V_Z\setminus K'} \right) \cap V_X$ such that $\alpha_u$ is non-zero. For any vertex $v \in \Lambda_{G'_1}^{V_Z \sm K'} \cap V_X$ such that $K' \varsubsetneq K$, $\alpha_u = 0\bmod \pi/2^{|V_Z|-2+\delta(|V_Z|-2)}$ by hypothesis of induction. If $|V_Z|-|K|=1$, then $u$ is a leaf and thus $\alpha_u = 0$. Else, by \cref{lemma:cond_angles}, $\sum_{w \in \Lambda_{G'_1}^{V_Z \sm K} \cap V_X}\alpha_w  = 0\bmod \dfrac{\pi}{2^{|V_Z|-2+\delta(|V_Z|-2)}}$, implying that $\alpha_u = 0\bmod \dfrac{\pi}{2^{|V_Z|-2+\delta(|V_Z|-2)}}$.
    \end{itemize}

    A vertex $u \in V_X$ such that for every $K \se V_Z$, $u \notin \Lambda_{G'_1}^{V_Z\setminus K}$, is an isolated vertex and thus $\alpha_u = 0$. Hence, for every $u \in V_X$, $\alpha_u=0\bmod \dfrac{\pi}{2^{|V_Z|-2+\delta(|V_Z|-2)}}$. By \cref{lemma:cond_angles}, for every vertex $v \in V_Z$, $\beta_v  = - \sum_{u \in N_{G'_1}(v)\cap V_X}\alpha_u = 0\bmod \dfrac{\pi}{2^{|V_Z|-2+\delta(|V_Z|-2)}}$. By \cref{lemma:less_than_half_Z}, $|V_Z|$ is upper bounded by $\lfloor n/2 \rfloor$, where $n$ is the order of $G_1$ and $G_2$. Thus, the angles are multiples of $\pi/2^{\lfloor n/2 \rfloor-1}$ ($\pi/2^{\lfloor n/2 \rfloor-2}$ in the case $|V_Z|\gs 3$). By \cref{lemma:implements_lc}, $\bigotimes_{u\in V_X}X(\alpha_u)\bigotimes_{v\in V_Z} Z(\beta_v)$ implements an $(\lfloor n/2 \rfloor-1)$-local complementation.    
\end{proof}

%\section{Proof of \cref{prop:lu-lc}}
\section{Proof of Proposition \ref{prop:lu-lc}}
\label{app:proof:lu-lc}

We begin by proving the following lemma.

\begin{lemma} \label{lemma:lu-lc}
    $LU \Leftrightarrow LC$ holds for graphs in standard form if and only if any $r$-local complementation over the vertices of type X can be implemented by local complementations.
\end{lemma}

\begin{proof}
    Let $G_1$ be a graph in standard form and $G_2$ be an arbitrary graph such that $\ket{G_1} =_\text{LU} \ket{G_2}$. By \cref{thm:LU_imply_LCr}, $G_1$ and $G_2$ are $(\lfloor n/2 \rfloor -1)$-locally equivalent, where $n$ is the order of the graphs. 
    By \cref{prop:standardform}, there exists $G'_2$ in standard form locally equivalent to $G_2$. By \cref{lemma:standardform_LCr_lc}, there exist $G''_2$ in standard form\footnote{$G''_2$ is in standard form because it is obtained from $G'_2$ by local complementations on vertices of type $\bot$.} locally equivalent to $G_2$ such that $G_1$ and $G''_2$ are related by a single $(\lfloor n/2 \rfloor -1)$-local complementation over the vertices of type X. 
    If the $(\lfloor n/2 \rfloor -1)$-local complementation can be implemented using local complementations, then $G_1$ and $G''_2$ are locally equivalent, so are $G_1$ and $G_2$.

    On the contrary, let $G_2$ be the result of an $r$-local complementation on a graph $G_1$. In particular, $\ket{G_1}$ and $\ket{G_2}$ are LU-equivalent. Hence, if $LU \Leftrightarrow LC$ holds for $G_1$, then $G_1$ and $G_2$ are locally equivalent, i.e. the $r$-local complementation can be implemented by local complementations.    
\end{proof}

\lulc*

\begin{proof}
    Any graph can be put in standard form by means of local complementations by \cref{prop:standardform}. \cref{lemma:lu-lc} along with the fact that the $LU \Leftrightarrow LC$ property is invariant by LC$_1$-equivalence, proves the proposition.
\end{proof}

\section{A graph where each vertex is either a leaf or adjacent to a leaf is in standard form}
\label{app:leaf_standard_form}

Let $G$ be a graph where each vertex is either of leaf of adjacent to a leaf. We suppose without loss of generality that $G$ is connected, as $LU \Leftrightarrow LC$ holds for a graph if and only if $LU \Leftrightarrow LC$ holds for its connected components. Moreover, we suppose without loss of generality that the graph $G$ is of order at least 3, else the result is trivial.

Let us partition the vertex set into $L$ "the leafs" and $P$ "the parents": $|L| \gs |P|$, any vertex in $L$ is adjacent to exactly one vertex in $P$, and any vertex in $P$ is adjacent to at least one vertex in $L$. Let us prove that the vertices in $L$ are of type X while the vertices in $P$ are of type Z. For any vertex $u \in L$, $\{u\} \cap N_G(u)$ is a minimal local set of dimension 1 generated by $\{u\}$ (as the unique neighbour of $u$ is not of degree 1). Moreover, a minimal local set cannot be generated by a set containing a vertex in $P$. Indeed, if a local set is generated by a set containing a vertex $v \in P$ related to a vertex $u\in L$, the local set is not minimal as it strictly contains the minimal local set $\{u,v\}$. Thus, every minimal local set is generated by a set in $L$. 

Fix any ordering of the vertices where for any leaf $u$ related to a vertex $v$, $u \prec v$. Then, $G$ is in standard form.

\section{$C_{t,k}$ and $C'_{t,k}$ are in standard form.}
\label{app:c_standard_form}

Fix two integers $k$ and $t$ such that $k \gs 3$ is odd and $t\gs k+2$. We prove that each vertex in $\binom{[1,t]}{k}$ is of type X and each vertex in $[1,t]$ is of type Z in $C_{t,k}$. First, let us prove by contradiction that a minimal local set cannot be generated by a set containing a vertex in $[1,t]$. Suppose that a set $D \se V$ such that $D \cap [1,t] \neq \emptyset$, generates a minimal local set $L$ in $C_{t,k}$. By symmetry, we suppose without loss of generality that $D \cap [1,t] = [1,m]$ for some $m\in[1,t]$. Let us exhibit a local set strictly contained in $L$, implying a contradiction with the fact that $L$ is a minimal local set:
    \begin{itemize}
        \item If $m=1$ i.e. $D \cap [1,t] = \{1\}$: $L$ contains every vertex of the form $\{\{1\}\cup x\}$ for some $x\in \binom{[2,t]}{k-1}$. Let $D' = \{\{1\}\cup x ~|~ x \in \binom{[2,k+1]}{k-1}\}$. $Odd_{C_{t,k}}(D') = \{1\}$, so $D' \cup Odd_{C_{t,k}}(D')$ is a local set contained strictly in $L$ ($L$ contains other vertices, e.g. the vertex $\{1,4,5,\cdots,k+2\}$).
        \item If $m\gs 2$: let $m_\text{odd}$ be the minimum between $k$ and the largest odd integer strictly less than $m$. Let $D' = \{\{1,\cdots,m_\text{odd},m+1,\cdots,m+k-m_\text{odd}\},\{2,\cdots,m_\text{odd}+1,m+1,\cdots,m+k-m_\text{odd}\}\}$. $Odd_{C_{t,k}}(D') = \{1,m_\text{odd}+1\}$, so $D' \cup Odd_{C_{t,k}}(D')$ is a local set contained strictly in $L$ ($L$ contains other vertices, e.g. the vertex $\{1,3,\cdots,m_\text{odd}+1,m+1,\cdots,m+k-m_\text{odd}\}$).
    \end{itemize} 
    Thus, every minimal local set is generated by a set in $\binom{[1,t]}{k}$. As every vertex is covered by a minimal local set by \cref{thm:MLScover}, it follows that each vertex in $\binom{[1,t]}{k}$ is of type X and each vertex in $[1,t]$ is of type Z. The proof is the same for $C'_{t,k}$ as $C_{t,k}$ and $C'_{t,k}$ differ only by edges between vertices of $[1,t]$ thus for any set $D \se V$, $(D \cup Odd_{C_{t,k}}(D)) \cap \binom{[1,t]}{k} = (D \cup Odd_{C'_{t,k}}(D)) \cap \binom{[1,t]}{k}$.

%\section{Proof of \cref{thm:existence_Ctk}}
\section{Proof of Theorem \ref{thm:existence_Ctk}}
\label{app:proof:existence_Ctk}

We begin by proving the following technical lemma.

\begin{lemma}
    Given any integers $s \ls m$, $\binom{m}{s} = 2^{w(s)+w(m-s)-w(m)}q$ where $q$ is an odd integer and $\ham x$ denotes the Hamming weight of the integer $x$.
\label{lemma:paritybinom}
\end{lemma}

\begin{proof}
    
    We prove by induction that for any $n \gs 0$, $n! = 2^{n-\ham n}q$ where $q$ is an odd integer. The property is trivially true for $n=0$. Given $n>0$, by induction hypothesis, there exists $q=1\bmod 2$ s.t. $n! = n (n-1)! = n2^{n-1-w(n-1)}q=2^kq'2^{n-1-w(n-1)}q= 2^{n-1+k-w(n-1)}qq'$, where $n=2^kq'$ and $q'=1\bmod 2$. Notice that $w(n)=w(q')$, as the binary notation of $n$ is nothing but the binary notation of $q'$ concatenated with $k$ zeros. Moreover, $w(n-1)=w(q')-1+k = w(n)-1+k$ as $n-1= 2^k(q'-1)+ \sum_{j=0}^{k-1}2^j$. Therefore, $n!=2^{n-w(n)}qq'$. 
   
    As a consequence there exist three odd integers $q_1$, $q_2$ and $q_3$ such that $$\binom{m}{s} = \dfrac{m!}{s! (m-s)!}=\dfrac{2^{m-w(m)}q_1}{2^{s-w(s)}q_2 2^{m-s-w(m-s)}q_3} = 2^{w(s)+w(m-s)-w(m)}\dfrac{q_1}{q_2 q_3}$$
\end{proof}

\existenceCtk*

\begin{proof}
    \cref{lemma:cond_lcr} along with \cref{lemma:cond_not_lcr} imply that for any odd $k \gs 3$ and $t\gs k+2$, $C_{t,k}$ and $C'_{t,k}$ are $r$-locally equivalent but not $(r-1)$-locally equivalent if
    \begin{equation*}
        \begin{split}
            \binom{t}{2} & = 1\bmod 2\\  
            \binom{k}{2} & = 0\bmod2^{r-1}\\ 
            \binom{t - 2}{k - 2} & = 2^{r-1}\bmod 2^r\\   
            \forall i \in [ 1, r-1 ]\text{, ~~}\binom{t - i -2}{k - i -2} & = 0\bmod 2^{r-i}                     
        \end{split}
    \end{equation*}  

    By \cref{lemma:paritybinom}, it is equivalent to prove the following equations:

    \begin{equation*}
        \begin{split}
            & \ham{t} - \ham{t-2}  = 1\\
            & \ham{k-2} - \ham{k}  \gs r-2 \\
            & \ham{k-2}+\ham{t-k}-\ham{t-2} =r -1\\   
            \forall i \in [ 1, r-1 ]\text{, ~~}& \ham{k-2 -i}+\ham{t-k} -\ham{t-2-i} \geqslant r -i            
        \end{split}
    \end{equation*}

    Let us prove that these equations are satisfied for $k = 2^r +1$ and $t = 2^r + 2^{\lceil \log_2(r+2) \rceil} -1$. We will make use of the following property of the Hamming weight: if $x\ls 2^z-1$, then $\ham{2^z-1-x}=z-\ham{x}$. In particular, for any $x \ls r+1 \ls 2^{\lceil \log_2(r+2) \rceil}-1$, $\ham{t-x}=\lceil \log_2(r+2) \rceil +1-\ham{x}$.
    \begin{itemize}
        \item $\ham{t}=\lceil \log_2(r+2) \rceil +1$, $\ham{t-2}=\lceil \log_2(r+2) \rceil$, thus $\ham{t} - \ham{t-2}  = 1$.
        
        \item $\ham{k}=2$ and $\ham{k-2}=\ham{2^r-1}=r$, thus $\ham{k-2} - \ham{k} = r-2$.
        
        \item $\ham{t-k} = \ham{2^{\lceil \log_2(r+2) \rceil}-2} = \lceil \log_2(r+2) \rceil -1$. Thus, $\ham{k-2}+\ham{t-k}-\ham{t-2}= r + \lceil \log_2(r+2) \rceil -1 - \lceil \log_2(r+2) \rceil = r -1$.
        
        \item Let $i \in [ 1, r-1 ]$. $\ham{k-2-i} = \ham{2^r-1-i}=r-\ham{i}$ and $\ham{t-2-i} = \ham{t -(i+2)} =  \lceil \log_2(r+2) \rceil +1 - \ham{i+2}$. Thus, $\ham{k-2 -i}+\ham{t-k} -\ham{t-2-i} = r-\ham{i} + \lceil \log_2(r+2) \rceil -1 - \left(\lceil \log_2(r+2) \rceil +1 - \ham{i+2}\right) = r - \left(\ham{i} - \ham{i+2} + 2\right)$. It is easy to check that $\ham{i} - \ham{i+2} + 2 \ls i$. Thus, $\ham{k-2 -i}+\ham{t-k} -\ham{t-2-i} \gs r -i$.
    \end{itemize}

Hence, $C_{t,k}$ and $C'_{t,k}$ are $r$-locally equivalent but not $(r-1)$-locally equivalent. By \cref{thm:LCr_lc}, $\ket{C_{t,k}}$ and $\ket{C'_{t,k}}$ are LC$_{r}$-equivalent but not LC$_{r-1}$-equivalent.    
\end{proof}
\end{document}